\pdfoutput=1
\documentclass[journal]{IEEEtran}

\usepackage{amsmath, amssymb, amsfonts, bm, mathtools}
\usepackage{amsthm}

\usepackage{nicematrix}
\usepackage{arydshln} 
\usepackage{multirow}

\newtheorem{corollary}{Corollary}

\newtheorem{proposition}{Proposition}

\newtheorem{remark}{Remark}
\usepackage[version=4]{mhchem}
\usepackage{chemfig}

\usepackage{algorithm}
\usepackage{algorithmic}

\usepackage{graphicx}
\usepackage{subfigure}
\usepackage{svg}
\usepackage{tikz}
\usepackage[american]{circuitikz}
\usetikzlibrary{
  arrows.meta, 
  decorations.markings, 
  circuits, 
  automata,
  arrows, 
  positioning, 
  calc
}
\usepackage{booktabs}
\usepackage{makecell}

\usepackage{textcomp}
\usepackage{soul}       
\usepackage{verbatim}
\usepackage{xcolor} 

\usepackage{cite}
\usepackage[hidelinks]{hyperref}
\usepackage[capitalise,noabbrev]{cleveref}
\Crefformat{figure}{#2Fig.~#1#3}
\Crefformat{equation}{#2Eq.~(#1)#3}
\crefname{equation}{Eq.}{Eqs.}
\crefname{figure}{Fig.}{Figs.}

\usepackage{xr}
\usepackage{xr-hyper} 

\usepackage[
  range-phrase=--,
  per-mode=symbol-or-fraction,
  range-units=single,
  list-units=single,
  detect-all,
  list-final-separator={, and },
]{siunitx}

\usepackage{acro}
\usepackage{mfirstuc}

\DeclareAcronym{MC}{short=MC, long=molecular communication}
\DeclareAcronym{MCvD}{short=MCvD, long=molecular communication via diffusion}
\DeclareAcronym{IM}{short=IM, long=information molecule, short-plural=s, long-plural=s}
\DeclareAcronym{TX}{short=TX, long=transmitter, short-plural=s, long-plural=s}
\DeclareAcronym{RX}{short=RX, long=receiver, short-plural=s, long-plural=s}
\DeclareAcronym{IoBNT}{short=IoBNT, long=Internet of Bio-Nano Things}
\DeclareAcronym{ssDNA}{short=ssDNA, long=single-stranded DNA}
\DeclareAcronym{cDNA}{short=cDNA, long=complementary DNA}

\DeclareAcronym{EM}{short=EM, long=electromagnetic}
\DeclareAcronym{FSO}{short=FSO, long=free-space optical}

\DeclareAcronym{ML}{short=ML, long=maximum-likelihood}
\DeclareAcronym{LS}{short=LS, long=least-squares}
\DeclareAcronym{NNLS}{short=NNLS, long=non-negative least squares}
\DeclareAcronym{MAP}{short=MAP, long=maximum a posteriori}
\DeclareAcronym{MMSE}{short=MMSE, long=minimum mean squared error}
\DeclareAcronym{MIMO}{short=MIMO, long=multiple-input multiple-output}
\DeclareAcronym{CDMA}{short=CDMA, long=code division multiple access}

\DeclareAcronym{MSE}{short=MSE, long=mean squared error}
\DeclareAcronym{FIM}{short=FIM, long=Fisher information matrix}
\DeclareAcronym{CRB}{short=CRB, long=Cramér-Rao bound}

\DeclareAcronym{SNR}{short=SNR, long=signal-to-noise ratio}
\DeclareAcronym{SINR}{short=SINR, long=signal-to-interference-plus-noise ratio}

\DeclareAcronym{LTI}{short=LTI, long=linear time-invariant}

\DeclareMathOperator{\Var}{Var}
\DeclareMathOperator{\Cov}{Cov}
\DeclareMathOperator{\E}{\mathbb{E}}

\usepackage{standalone}

\usepackage{listofitems} 
\usepackage[outline]{contour}
\contourlength{1.4pt}

\colorlet{myred}{red!80!black}
\colorlet{myblue}{blue!80!black}
\colorlet{mygreen}{green!60!black}
\colorlet{myorange}{orange!70!red!60!black}
\colorlet{mydarkred}{red!30!black}
\colorlet{mydarkblue}{blue!40!black}
\colorlet{mydarkgreen}{green!30!black}

 \usepackage{bbm}

\usepackage{titlesec}
\titlespacing\section{0pt}{6pt plus 2pt minus 2pt}{4pt plus 2pt minus 2pt}
\titlespacing\subsection{0pt}{4pt plus 2pt minus 2pt}{3pt plus 2pt minus 2pt}

\usepackage{eso-pic}
\AddToShipoutPictureFG*{%
  \AtPageUpperLeft{%
    \put(\LenToUnit{0.5\paperwidth},\LenToUnit{-0.25in}){%
      \makebox(0,0)[t]{\parbox{0.9\textwidth}{\centering\scriptsize
        \copyright~2026 IEEE. Personal use of this material is permitted. Permission from IEEE must be obtained for all other uses, in any current or future media, including reprinting/republishing this material for advertising or promotional purposes, creating new collective works, for resale or redistribution to servers or lists, or reuse of any copyrighted component of this work in other works.
        This article has been accepted for publication in the IEEE Internet of Things Journal.}}}}}

\begin{document}
\bstctlcite{IEEEexample:BSTcontrol}

\title{Markov State-Space Modeling and Channel Characterization for DNA-Based Molecular Communication}
\author{Ruifeng~Zheng, Zhihan Xu, Veronika Volkova, Pengjie Zhou, Martín Schottlender, Juan A. Cabrera, \\
Frank H.\,P. Fitzek, and Pit Hofmann
\thanks{Part of this work was presented at IEEE Global Communications Conference 2025~\cite{zheng2025DNA-Based}. R.~Zheng, Z.~Xu, V.~Volkova, P.~Zhou, M.~Schottlender, J.\,A.~Cabrera, F.\,H.\,P.~Fitzek and P.~Hofmann are with the Deutsche Telekom Chair of Communication Networks, Dresden University of Technology, Germany; J.~Cabrera, F.~Fitzek and P.~Hofmann are also with the Centre for Tactile Internet with Human-in-the-Loop (CeTI), Dresden, Germany, email: \{ruifeng.zheng, zhihan.xu, veronika.volkova, pengjie.zhou, martin.schottlender, juan.cabrera, frank.fitzek, pit.hofmann\}@tu-dresden.de.
This work was supported by the German Research Foundation (DFG) as part of Germany's Excellence Strategy -- EXC 2050/2 -- Cluster of Excellence ``Centre for Tactile Internet with Human-in-the-Loop'' (CeTI) of Technische Universität Dresden under project ID 390696704, and by the Federal Ministry of Research, Technology and Space of Germany (BMFTR) as part of the research program Communication Systems ``Souverän. Digital. Vernetzt.''. Joint project 6G-life, project id number 16KIS2413K, the program “Verbundprojekt: Disruptive Kommunikationsparadigmen für technologische Souveränität, Resilienz und Shared Prosperity - Translation in Industrie und Aufbau innovativer Technologiedemonstratoren - CommUnity,” project id number 16KISS012K, and the project IoBNT with grant number 16KIS1994.}
}
\maketitle

\begin{abstract}
DNA-based molecular communication with microarray-style reception under reversible hybridization is of interest for emerging DNA-enabled sensing and Internet of Bio-Nano Things (IoBNT) systems. In such channels, the bound-state observation exhibits both inter-symbol interference (ISI) and colored counting noise. To capture these effects in a communication-oriented form, we develop a Markov state-space framework based on a voxelized reaction-diffusion model, in which a block-structured transition matrix describes molecular transport and binding/unbinding dynamics. For the microarray specialization, this representation yields the channel impulse response (CIR), the equilibrium gain, and a settling-time-based characterization of the effective channel memory. Building on the resulting symbol-rate observation model for on-off keying, we derive a grouped-binomial counting model and obtain a closed-form expression for the covariance of the counting noise. Based on these statistics, we further develop a differential-threshold detector and a finite-memory decision-feedback equalizer and benchmark them against a symbol-by-symbol maximum a posteriori detector with decision feedback (MAP-DF). Numerical results validate the theoretical correlation behavior and show that the relative performance of the proposed receivers depends strongly on the channel-memory regime. 
\end{abstract}

\begin{IEEEkeywords}
Internet of Bio-Nano Things (IoBNT), DNA-based molecular communication, microarray receiver, Markov state-space model, channel characterization, colored noise, inter-symbol interference, low-complexity detection.
\end{IEEEkeywords}

\section{Introduction}
\label{sec:intro}

\IEEEPARstart{M}{olecular} communication (MC) enables information exchange at the nanoscale, where conventional electromagnetic signaling is often infeasible or inefficient~\cite{yang2020comprehensive,bi2021survey,akan2023internet,kong2023survey,aktas2024odor,wen2025autonomous}. As the Internet of Things (IoT) continues to evolve toward bio-integrated and nanoscale systems, MC is increasingly recognized as a promising enabler for the Internet of Bio-Nano Things (IoBNT)~\cite{liu2021dna,akan2023internet,kong2023survey}. Among different MC modalities, DNA-based MC is particularly attractive due to the programmability, specificity, and stability of DNA hybridization~\cite{liu2021dna,gomez2024dna,preuss2024sequencing}. This motivates microarray-style receivers (RXs), in which transmitted information is inferred from the population of surface-bound DNA molecules formed under reversible hybridization dynamics~\cite{chan1995biophysics,hassibi2005biological,liu2021dna,wang2026dna}.

A key feature of microarray-style DNA reception is that the RX observes an inventory-like bound-state population rather than cumulative molecular arrivals~\cite{chan1995biophysics,hassibi2005biological}. Under reversible binding and unbinding, previously bound molecules may persist across multiple symbol intervals, which induces inter-symbol interference (ISI) and yields a long \emph{effective channel memory}~\cite{jamali2019channel,zheng2026molecular}. At the same time, the same molecules may contribute to multiple symbol-time observations, which gives rise to temporally correlated, i.e., colored, \emph{counting noise}. Therefore, a communication-oriented model for DNA microarray channels should characterize not only the mean response, but also the symbol-rate covariance of the counting noise, which is directly relevant to receiver design in DNA-enabled sensing and IoBNT systems.

\subsection{Related Works and Main Challenges}
\label{subsec:intro_related_works}

Existing MC channel models have largely focused on diffusion with absorbing or partially absorbing receivers~\cite{yilmaz2014three,yilmaz2014arrival,cao2020optimal,ferrari2022channel,sabu2022channel,chen2025channel}. While these models provide useful insights, they do not capture the reversible, inventory-like reception mechanism of DNA microarray channels. More recent communication-oriented studies have also considered equilibrium signaling strategies and practical biosensor-assisted microfluidic receivers~\cite{akdeniz2020equilibrium,abdali2024frequency}, but these models still do not represent spatially confined DNA hybridization with planar microarray-like RXs and reversible surface binding. Meanwhile, DNA-based MC has been investigated from both theoretical and experimental perspectives, including studies on channel capacity, communication protocols, parallel transmission architectures, microscale modulators, and synthetic DNA receivers~\cite{sun2019channel,bilgin2018dna,wang2023novel,luo2020small,zhang2023hardware}. Related IoBNT-oriented studies have also highlighted the broader role of molecular communication in bio-nano networking and channel modeling~\cite{al2022intrabody,yue2024bio,jing2024energy,dai2025cross}. However, these studies generally do not provide a communication-oriented model for spatially confined DNA hybridization channels with planar microarray-like receivers and reversible surface binding. In parallel, the biosensor and microarray literature has long studied diffusion-reaction coupling and reversible surface chemistry in bounded geometries~\cite{pappaert2003diffusion,squires2008making,chan1995biophysics}, but these models are rarely formulated in a communication-oriented form that directly yields symbol-rate channel metrics and receiver-relevant noise statistics.

Another major challenge lies in noise modeling. Standard MC RX models often adopt independent Poisson counting or temporally white Gaussian approximations~\cite{yilmaz2014arrival,li2016local,cao2020optimal}. Related communication-oriented receiver models have also considered Poisson arrival uncertainty together with receiver-side baseline noise and nonlinear activation, as in olfaction-inspired cross-reactive arrays~\cite{jamali2023olfaction}. Prior studies on reactive receivers have further recognized that ligand-receptor binding itself introduces additional stochasticity beyond simple diffusion-counting noise~\cite{pierobon2011noise}. However, these models do not directly capture the symbol-rate-colored counting noise arising in reversible DNA reception. In DNA microarray channels, the symbol-rate observation is affected not only by ISI in the mean response, but also by temporally correlated counting noise induced by reversible molecular persistence at the receiver. Hence, both mean-domain memory and noise-domain memory should be explicitly characterized.

In summary, existing studies leave three key gaps for DNA-based MC. First, conventional absorbing or partially absorbing receiver models do not capture the reversible hybridization process, where molecules may bind, unbind, and rebind to the receiver. Second, existing biosensor and microarray reaction-diffusion models are generally not formulated in a symbol-rate communication framework that directly yields the channel impulse response (CIR), channel memory, and observation statistics. Third, commonly used independent counting or white noise approximations fail to capture the temporally colored counting noise caused by molecular persistence. Together, these limitations hinder the development of an IoBNT communication model that can simultaneously capture the underlying biochemical dynamics and provide communication-oriented performance characterization.

\subsection{Main Contributions and Structure}
\label{subsec:intro_approach}

Our prior conference paper~\cite{zheng2025DNA-Based} presented a preliminary Markov-based model for DNA MC, mainly focusing on the transition-matrix construction, the CIR, the equilibrium distribution, the settling time, and a basic on-off keying (OOK) detection example. The present paper substantially extends that preliminary conference work in both modeling depth and communication-level analysis. In particular, we develop a more general block-structured Markov state-space framework, clarify the physical assumptions behind the reactive-boundary discretization, derive symbol-rate grouped-binomial observation statistics, characterize the temporal covariance of the counting noise, and evaluate receiver performance under different channel-memory regimes. We also strengthen the validation and benchmarking by including particle-based simulation (PBS) results for the channel response, Monte Carlo validation of the colored noise covariance, a Gaussian-approximation validation in the appendix, and a symbol-by-symbol maximum a posteriori detector with decision feedback (MAP-DF) as a receiver benchmark.

The main contributions of this paper are summarized as follows:
\begin{itemize}
    \item We develop a block-structured Markov state-space framework for DNA-based MC with reversible hybridization, providing a communication-oriented representation of a reactive microarray receiver.
    \item We specialize the framework to a tractable one-dimensional microarray channel and derive key channel characteristics, including the CIR, equilibrium gain, and effective channel memory, with validation via particle-based simulation.
    \item We derive a symbol-rate grouped-binomial observation model for OOK signaling and obtain a closed-form covariance expression that characterizes the resulting temporally colored counting noise.
    \item We develop low-complexity differential receivers, including a differential-threshold detector and finite-memory decision-feedback equalizers (DFEs), and compare them with a MAP-DF detector to study the tradeoff between ISI mitigation and receiver complexity.
\end{itemize}

The remainder of this paper is organized as follows. \Cref{sec:general_framework} presents the general Markov state-space framework. \Cref{sec:channel_specialization} specializes this framework to the DNA microarray channel and derives the corresponding analytical channel characterization. \Cref{sec:io_noise} develops the symbol-rate observation model and characterizes the covariance of the counting noise. \Cref{sec:detection} presents the proposed detectors. \Cref{sec:numerical_results} provides numerical results, and \Cref{sec:conclusion} concludes the paper and outlines future research directions.

\textit{Notation:} $\mathbb{R}^{m\times n}$ denotes the set of all $m\times n$ real-valued matrices, and $\mathbb{Z}_{\ge 0}$ denotes the set of nonnegative integers. Boldface uppercase and lowercase letters denote matrices and vectors, respectively. The superscript $(\cdot)^\top$ denotes transpose, and $\mathbbm{1}\{\cdot\}$ denotes the indicator function. Moreover, $\mathcal{N}(\mu,\sigma^2)$ and $\mathcal{B}(n,p)$ denote the Gaussian and binomial distributions, respectively.
\section{System Model and General Markov Framework}
\label{sec:general_framework}

In this section, we develop a general Markov state-space framework for DNA-based MC with reversible hybridization at a reactive receiver boundary. The intended application scenario is a confined microfluidic or microarray-style DNA reception chamber, motivated by the confined biochip and affinity-based biosensor settings studied in~\cite{pappaert2003diffusion,hassibi2005biological}. In this scenario, the non-receiver boundaries represent non-reactive walls and are therefore modeled as reflecting/no-flux boundaries, while the receiver boundary is a probe-functionalized microarray surface with reversible binding and unbinding dynamics. By voxelizing the underlying reaction--diffusion system, we obtain a finite-state Markov chain whose block-structured transition matrix captures molecular transport and receiver-side hybridization in a unified form.

\begin{figure}
    \centering
    \includegraphics[width=0.6\linewidth]{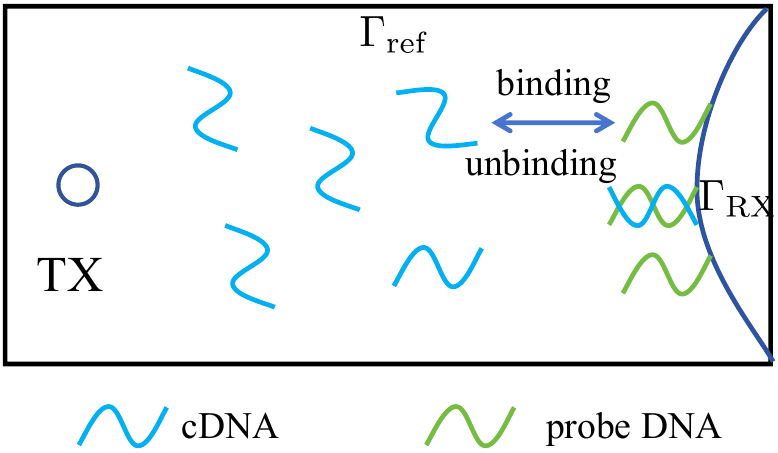}
    \caption{Physical model of the considered DNA microarray channel. cDNA molecules released from the TX propagate in a bounded domain $\Omega$ toward a reactive receiver boundary $\Gamma_{\mathrm{RX}}$ with reversible hybridization, while the remaining boundary $\Gamma_{\mathrm{ref}}$ is reflecting. The RX observation is the number of bound cDNA molecules.}
    \label{fig:system_model}
\end{figure}

\subsection{Physical Model and Discrete-State Representation}
\label{subsec:physical_and_states}

We consider a DNA-based MC system in a bounded domain $\Omega$, where the transmitter (TX) releases complementary DNA (cDNA) molecules that propagate toward a microarray RX. The RX occupies a reactive subset of the boundary, $\Gamma_{\mathrm{RX}} \subset \partial\Omega$, while the remaining boundary $\Gamma_{\mathrm{ref}} \triangleq \partial\Omega \setminus \Gamma_{\mathrm{RX}}$ is reflecting; see~\cref{fig:system_model}. The concentration of free cDNA molecules at position $\mathbf r=(x,y,z)$ and time $t$ is denoted by $c(\mathbf r,t)$. Its spatiotemporal evolution is governed by the advection-diffusion equation
\begin{equation}
    \frac{\partial c}{\partial t}
    = \nabla \cdot \left( D \nabla c - \mathbf{v}\, c \right),
    \qquad \mathbf{r}\in\Omega,
    \label{eq:pde_cd}
\end{equation}
where $D$ is the diffusion coefficient and $\mathbf{v}=(v_x,v_y,v_z)$ is the drift velocity vector.

At the reactive boundary $\Gamma_{\mathrm{RX}}$, free cDNA information molecules reversibly hybridize with surface-immobilized probe DNAs. In this work, each immobilized probe DNA is modeled as a receptor site of the RX. We denote by $\sigma_{\mathrm p}$ the surface density of available probe sites on $\Gamma_{\mathrm{RX}}$, and by $c_{\mathrm b}$ the surface density of bound cDNAs. The coupling between transport in the medium and surface reactions is described by the Robin-type boundary condition~\cite{squires2008making,pappaert2003diffusion}
\begin{equation}
    -\left( D\nabla c - \mathbf{v}\,c \right)\cdot \mathbf{n}
    = k_{\mathrm{on}}\,\sigma_{\mathrm p}\,c -k_{\mathrm{off}}\,c_{\mathrm{b}},
    \quad \mathbf{r}\in\Gamma_{\mathrm{RX}},\ t>0,
    \label{eq:robin_bc}
\end{equation}
where $\mathbf{n}$ is the outward unit normal vector, and $k_{\mathrm{on}}$ and $k_{\mathrm{off}}$ are the binding and unbinding rate constants, respectively. The same reaction flux governs the bound-state dynamics,
\begin{equation}
    \frac{\partial c_{\mathrm{b}}}{\partial t}
    = k_{\mathrm{on}}\,\sigma_{\mathrm p}\,c\big|_{\Gamma_{\mathrm{RX}}}
    - k_{\mathrm{off}}\,c_{\mathrm{b}}.
    \label{eq:cb_ode}
\end{equation}
On the reflecting boundary $\Gamma_{\mathrm{ref}}$, we impose
\begin{equation}
    \left( D\nabla c - \mathbf{v}\,c \right)\cdot \mathbf{n} = 0,
    \qquad \mathbf{r}\in\Gamma_{\mathrm{ref}},\ t>0.
    \label{eq:reflect_bc}
\end{equation}

To obtain a tractable communication-oriented model, we approximate the reaction-diffusion system in~\cref{eq:pde_cd,eq:robin_bc,eq:cb_ode,eq:reflect_bc} by a discrete-state Markov chain using spatial voxelization and a sufficiently small sampling interval $\Delta t$. From a single-molecule perspective, this yields a probability-preserving random walk in which, during each interval $\Delta t$, a molecule may diffuse to a neighboring voxel, bind or unbind at the receiver boundary, or remain in its current state. The step size $\Delta t$ is chosen sufficiently small so that all one-step transition probabilities are nonnegative and sum to one. Unless stated otherwise, we focus on the pure-diffusion case in the subsequent Markov construction.

Throughout the stochastic channel and RX modeling, we adopt a non-saturating probe assumption, i.e., the number of released molecules is sufficiently small relative to the available probe-site capacity such that probe-site saturation can be neglected. Physically, this means that the bound fraction remains well below full probe-site occupancy, so the effective density of available probe sites can be treated as approximately constant. Under this assumption, the boundary reaction remains linear in the free-molecule concentration, different molecules evolve independently according to the same fixed single-molecule dynamics, and the total symbol-rate observation can be represented as the superposition of independent single-molecule contributions. Similar non-saturation assumptions have also been adopted in related studies, e.g.,~\cite{hassibi2005biological,squires2008making,pierobon2011noise,kuscu2016physical}.
High-occupancy or saturation-limited regimes, in which competition among probe sites couples molecular trajectories, are outside the scope of the present work.

We define two disjoint state subsets. The \emph{free-state set} contains the volume voxels in $\Omega$,
\begin{equation}
    \mathcal{S}_{\mathrm{free}}=\{s_{1}, s_{2}, \ldots, s_{N_\mathrm{f}}\},
    \label{eq:S_free}
\end{equation}
where $N_\mathrm{f}$ is the number of free-diffusion states. The \emph{bound-state set} contains states associated with reactive boundary subregions on $\Gamma_{\mathrm{RX}}$,
\begin{equation}
    \mathcal{S}_{\mathrm{bound}} =\{s_{N_\mathrm{f}+1}, \ldots, s_{N_\mathrm{f}+N_\mathrm{b}}\},
    \label{eq:S_bound}
\end{equation}
where $N_\mathrm{b}$ is the number of bound states. The complete state space is
\begin{equation}
    \mathcal{S}=\mathcal{S}_{\mathrm{free}}\cup\mathcal{S}_{\mathrm{bound}},
\end{equation}
with $N=N_{\mathrm f}+N_{\mathrm b}$ states.

The state of a single information molecule at time $t=n\Delta t$ is denoted by $R[n]\in\mathcal{S}$. We model $\{R[n]\}$ as a time-homogeneous Markov chain with one-step transition probabilities
\begin{equation}
    P_{i,j}
    \triangleq
    \Pr\left\{ R[n+1]=s_i \mid R[n]=s_j \right\},
    \label{eq:one_step_transition}
\end{equation}
where $i$ denotes the target state and $j$ denotes the source state.

\subsection{Transition Matrix Construction}
\label{subsec:transition_matrix}

Stacking the one-step transition probabilities in~\cref{eq:one_step_transition} yields a column-stochastic transition matrix $\mathbf{P}\in\mathbb{R}^{N\times N}$, where each column sums to one. With the state ordering $(s_{1},\ldots,s_{N_{\mathrm{f}}},\,s_{N_{\mathrm{f}}+1},\ldots,s_{N_{\mathrm{f}}+N_{\mathrm{b}}})$, $\mathbf{P}$ admits the block form
\begin{equation}
    \mathbf{P}=
    \begin{bmatrix}
    \mathbf{Q} & \mathbf{U} \\
    \mathbf{B} & \mathbf{R}
    \end{bmatrix},
    \label{eq:P_blocks}
\end{equation}
where $\mathbf{Q}\in\mathbb{R}^{N_\mathrm{f}\times N_\mathrm{f}}$ describes transport among free states, $\mathbf{B}\in\mathbb{R}^{N_\mathrm{b}\times N_\mathrm{f}}$ describes binding transitions from free states into bound states, $\mathbf{U}\in\mathbb{R}^{N_\mathrm{f}\times N_\mathrm{b}}$ describes unbinding transitions back into free states, and $\mathbf{R}\in\mathbb{R}^{N_\mathrm{b}\times N_\mathrm{b}}$ describes retention within the bound-state set. Hence, the block structure separates transport in the medium from receptor kinetics at the receiver boundary.

\begin{figure}[t]
    \centering
    \subfigure[Interior voxel.]{
        \centering
        \includegraphics[width=0.45\columnwidth]{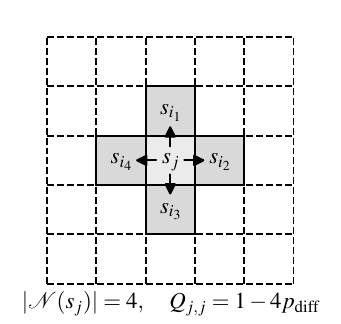}
        \label{fig:Q_neighborhood_interior}
    }\hfill
    \subfigure[Reflecting boundary.]{
        \centering
        \includegraphics[width=0.45\columnwidth]{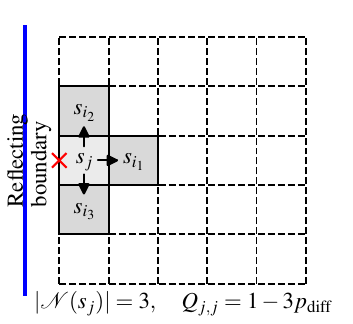}
        \label{fig:Q_neighborhood_reflecting}
    }
    \caption{Neighborhood set $\mathcal{N}(s_j)$ used to construct the diffusion block $\mathbf{Q}$ (2D example). (a) Interior voxel: $|\mathcal{N}(s_j)|=4$. (b) Voxel adjacent to a reflecting boundary: one diffusion direction is blocked, yielding $|\mathcal{N}(s_j)|=3$. The construction extends to 1D and 3D by replacing the interior neighborhood size with $2$ and $6$, respectively.}
    \label{fig:Q_neighborhood}
\end{figure}

To obtain a probability-preserving Markov discretization, we focus on a diffusion-dominated confined microfluidic or microarray-style setting and neglect macroscopic drift, i.e., we set $\mathbf v=\mathbf 0$. Under this assumption, on a voxel grid with spacing $\Delta x$, the per-step diffusion probability is defined as
\begin{equation}
p_{\mathrm{diff}}\triangleq\frac{D\,\Delta t}{(\Delta x)^2}.
\label{eq:pd}
\end{equation}
Although the present analysis assumes zero drift, nonzero drift can be incorporated by modifying the transport block $\mathbf Q$ using biased-random-walk or upwind transition probabilities that preserve column stochasticity.\footnote{For details on flow-aware Markov discretizations, including biased-random-walk and upwind constructions that preserve column stochasticity, see~\cite{zheng2026system}.} The resulting flow-aware transition matrix would generally change the CIR and memory characteristics. Modeling spatially varying hemodynamic profiles and rheological effects, however, is beyond the scope of this work.

Binding is possible only from free voxels adjacent to the reactive boundary; we denote this receiver-adjacent set by $\mathcal{V}_{\mathrm{RX}}\subseteq\mathcal{S}_{\mathrm{free}}$. 
To parameterize the corresponding free-to-bound transition probabilities, the probe population on each reactive surface patch is associated with the adjacent free voxel in $\mathcal{V}_{\mathrm{RX}}$. We use $c_{\mathrm p}$ to denote the effective volumetric concentration of available immobilized probe DNAs in this voxel. If the probes are specified by a probe surface density $\sigma_{\mathrm p}$, then assigning the probes on a reactive surface patch of area $A_{\mathrm{RX}}$ to the adjacent boundary voxel of volume $A_{\mathrm{RX}}\Delta x$ gives
\begin{equation}
c_{\mathrm p}
=
\frac{\sigma_{\mathrm p}A_{\mathrm{RX}}}
{A_{\mathrm{RX}}\Delta x}
=
\frac{\sigma_{\mathrm p}}{\Delta x}.
\label{eq:cp}
\end{equation}
Here, $[\sigma_{\mathrm p}]=\mathrm{mol\cdot m^{-2}}$, $[A_{\mathrm{RX}}]=\mathrm{m^2}$, and $[\Delta x]=\mathrm m$, so that $[c_{\mathrm p}]=\mathrm{mol\cdot m^{-3}}$, which may equivalently be expressed in $\mathrm M$ or $\mu\mathrm M$ after unit conversion. 

The binding and unbinding transitions are then modeled as first-order stochastic jump processes. A reaction with rate $\xi$ occurs within one time step $\Delta t$ with probability $1-\exp(-\xi\Delta t)$, which reduces to $\xi\Delta t$ when $\xi\Delta t\ll 1$~\cite{gillespie1977exact,gillespie2007stochastic,erban2007reactive}. Accordingly, a free cDNA molecule in the receiver-adjacent reactive voxel experiences the pseudo-first-order binding rate
\begin{equation}
\xi_{\mathrm{bind}}
=
k_{\mathrm{on}}c_{\mathrm p},
\end{equation}
yielding the one-step binding probability
\begin{equation}
p_{\mathrm{bind}}
\triangleq
1-\exp\left(-k_{\mathrm{on}}c_{\mathrm p}\Delta t\right)
\approx
k_{\mathrm{on}}c_{\mathrm p}\Delta t,
\label{eq:p_bind_general}
\end{equation}
where the approximation holds for $k_{\mathrm{on}}c_{\mathrm p}\Delta t\ll 1$. Similarly, a bound cDNA molecule experiences the first-order unbinding rate
\begin{equation}
\xi_{\mathrm{unbind}}
=
k_{\mathrm{off}},
\end{equation}
yielding the one-step unbinding probability
\begin{equation}
p_{\mathrm{unbind}}
\triangleq
1-\exp\left(-k_{\mathrm{off}}\Delta t\right)
\approx
k_{\mathrm{off}}\Delta t,
\label{eq:p_unbind_general}
\end{equation}
where the approximation holds for $k_{\mathrm{off}}\Delta t\ll 1$. The time step $\Delta t$ is chosen sufficiently small so that the resulting one-step transition probabilities remain in $[0,1]$ and the transition matrix remains column-stochastic.

For a free voxel $s_j$, denote its neighboring free-voxel set by $\mathcal{N}(s_j)\subseteq\mathcal S_{\mathrm{free}}$. Then, for $s_i,s_j\in\mathcal S_{\mathrm{free}}$, the diffusion block entries are
\begin{equation}
    Q_{i,j}=
    \begin{cases}
    p_{\mathrm{diff}}, & s_i\in\mathcal{N}(s_j),\\[2pt]
    1-|\mathcal{N}(s_j)|\,p_{\mathrm{diff}}, 
    & i=j,\ s_j\notin\mathcal V_{\mathrm{RX}},\\[2pt]
    1-|\mathcal{N}(s_j)|\,p_{\mathrm{diff}}-p_{\mathrm{bind}}, 
    & i=j,\ s_j\in\mathcal V_{\mathrm{RX}},\\[2pt]
    0, & \text{otherwise}.
    \end{cases}
    \label{eq:qij_general}
\end{equation}
Reflecting boundaries reduce $|\mathcal{N}(s_j)|$ and therefore increase the self-transition probability. To ensure nonnegative probabilities, $\Delta t$ is chosen such that
$|\mathcal{N}(s_j)|\,p_{\mathrm{diff}}\le 1$ for $s_j\notin\mathcal V_{\mathrm{RX}}$, and
$|\mathcal{N}(s_j)|\,p_{\mathrm{diff}}+p_{\mathrm{bind}}\le 1$ for $s_j\in\mathcal V_{\mathrm{RX}}$.
\Cref{fig:Q_neighborhood} illustrates this neighborhood construction in 2D.

If multiple bound states are used to represent different reactive subregions, binding transitions from a free voxel $s_j\in\mathcal{V}_{\mathrm{RX}}$ are distributed among bound states via nonnegative splitting coefficients $\{\alpha_{b,j}\}$ satisfying $\sum_{b=1}^{N_{\mathrm b}}\alpha_{b,j}=1$. Accordingly,
\begin{equation}
    B_{b,j}=
    \begin{cases}
    p_{\mathrm{bind}}\,\alpha_{b,j}, & s_j \in \mathcal{V}_{\mathrm{RX}},\\[2pt]
    0, & \text{otherwise}.
    \end{cases}
    \label{eq:B_general}
\end{equation}
If the receiver is homogeneous, i.e., the immobilized probes have identical binding kinetics and are uniformly distributed over equivalent reactive subregions without spatial preference in local accessibility, then the uniform binding rule $\alpha_{b,j}=1/N_{\mathrm b}$ is reasonable when multiple bound states are used.

Similarly, upon unbinding from bound state $s_{N_{\mathrm f}+b}$, a molecule is released into adjacent free voxels according to coefficients $\{\beta_{i,b}\}$ satisfying
$\sum_{i:\,s_i\in\mathcal{V}_{\mathrm{RX}}}\beta_{i,b}=1$, i.e.,
\begin{equation}
    U_{i,b}=
    \begin{cases}
    p_{\mathrm{unbind}}\,\beta_{i,b}, & s_i \in \mathcal{V}_{\mathrm{RX}},\\[2pt]
    0, & \text{otherwise}.
    \end{cases}
    \label{eq:U_general}
\end{equation}
For the reverse transition, we use $\beta_{i,b}=1/|\mathcal{V}_{\mathrm{RX}}|$ as the uniform unbinding rule, which assumes that a released molecule enters the adjacent free voxels without directional bias. More general nonuniform $\alpha_{b,j}$ and $\beta_{i,b}$ can be used when probe density, reactive area, local geometry, or transport conditions are spatially nonuniform.

We neglect lateral transitions among bound states, so that a bound molecule either remains bound or unbinds. Hence,
\begin{equation}
    \mathbf{R}=\mathrm{diag}(1-p_{\mathrm{unbind}},\,\ldots,\,1-p_{\mathrm{unbind}}).
    \label{eq:R_diag}
\end{equation}
Together, \cref{eq:qij_general,eq:B_general,eq:U_general,eq:R_diag} ensure that $\mathbf P$ is column-stochastic.

\begin{remark}[Effect of Probe Coverage and Saturation]
\label{rem:probe_coverage_saturation}
Under the non-saturating probe assumption introduced in~\cref{subsec:physical_and_states}, the binding probability $p_{\mathrm{bind}}$ is constant and the transition matrix $\mathbf P$ is fixed. For a finite microarray, this assumption requires the available probe capacity, determined by the probe surface density $\sigma_{\mathrm p}$ and effective reactive area $A_{\mathrm{RX}}$, to be sufficiently large relative to the total number of released molecules. A quantitative probe-capacity assessment is provided in~\cref{sec:supp_probe_capacity} of the Supplementary Material. If the boundary is only partially covered by receptors, this can be incorporated by defining $\mathcal V_{\mathrm{RX}}$ only on probe-covered boundary-adjacent voxels. In contrast, if probe-site saturation becomes significant, the binding probability becomes occupancy-dependent, and the fixed transition matrix should be replaced by a time-varying matrix $\mathbf P(t)$.
\end{remark}

The transition matrix $\mathbf P$ therefore provides a compact representation of the physical channel environment. The transport block $\mathbf Q$ captures the diffusion geometry and boundary conditions, the binding block $\mathbf B$ represents the reactive coupling from receiver-adjacent free voxels to bound states, the unbinding block $\mathbf U$ describes molecular release from the receiver surface back to the medium, and $\mathbf R$ captures bound-state retention. Consequently, changes in the channel geometry, boundary conditions, probe distribution, or receiver type can be incorporated by modifying the corresponding blocks of $\mathbf P$. For a passive receiver, the RX mainly enters through the observation vector, whereas for the reactive microarray receiver considered here, the RX also affects the transition dynamics through the binding and unbinding blocks.

\subsection{State Evolution and Observation Mapping}
\label{subsec:state_evolution_observation}

Given the transition matrix $\mathbf P$, we track a single information molecule by its state-occupancy probabilities over $\mathcal S$. Specifically, the probability that the molecule is in state $s_i$ at time $t=k\Delta t$ is denoted by $x_i[k]$. Collecting these probabilities gives
\begin{equation}
    \bm{x}[k] \triangleq [x_{1}[k],\ldots,x_{N}[k]]^{\top}\in\mathbb{R}^{N},
    \label{eq:x_discrete}
\end{equation}
where $x_{i}[k]\triangleq \Pr\{R[k]=s_i\}$. By construction, $\bm x[k]\succeq \bm 0$ and $\|\bm{x}[k]\|_{1}=1$. Since $\mathbf{P}$ is column-stochastic, the state distribution evolves as
\begin{equation}
    \bm{x}[k+1]=\mathbf{P}\,\bm{x}[k],
    \label{eq:markov_update_disc_general}
\end{equation}
which yields
\begin{equation}
    \bm{x}[k]=\mathbf{P}^{k}\bm{x}[0]
    \label{eq:markov_solution_disc_general}
\end{equation}
for an arbitrary initial distribution $\bm{x}[0]$.

We define the set of observable states as $\mathcal{S}_{\mathrm{obs}}\subseteq\mathcal{S}$, and the corresponding observable-state indicator vector $\boldsymbol{o}\in\{0,1\}^{N}$ by
\begin{equation}
    o_i=
    \begin{cases}
    1, & s_i\in\mathcal{S}_{\mathrm{obs}},\\
    0, & \text{otherwise}.
    \end{cases}
    \label{eq:obs_vector_general}
\end{equation}
The resulting single-molecule observable-state probability is
\begin{equation}
    y[k] \triangleq \boldsymbol{o}^{\top}\bm{x}[k]
    = \Pr\{R[k]\in\mathcal S_{\mathrm{obs}}\},
    \label{eq:observation_mapping_discrete}
\end{equation}
where the equality follows directly from~\cref{eq:x_discrete,eq:obs_vector_general}. For microarray-style reception, we set $\mathcal S_{\mathrm{obs}}=\mathcal S_{\mathrm{bound}}$, so that $y[k]$ represents the probability that a single molecule is bound at time $k\Delta t$.

For notational convenience, we associate a continuous-time trajectory $\bm{x}(t)$ with the sampled process through $\bm{x}(k\Delta t)=\bm{x}[k]$. When $\Delta t$ is sufficiently small, we use the following continuous-time surrogate as a first-order approximation,
\begin{equation}
    \mathbf{A}\triangleq\frac{\mathbf{P}-\mathbf{I}}{\Delta t},
    \label{eq:A_def}
\end{equation}
and write
\begin{equation}
    \frac{d\bm{x}(t)}{dt}\approx \mathbf{A}\,\bm{x}(t), \qquad 
    \bm{x}(t)\approx \exp(\mathbf{A}t)\,\bm{x}(0),
    \label{eq:cts_approx}
\end{equation}
with the corresponding observation $y(t)\triangleq \boldsymbol{o}^{\top}\bm{x}(t)$.

\subsection{Channel Response and Key Characteristics}
\label{subsec:channel_characteristics}

Based on~\cref{eq:markov_solution_disc_general,eq:observation_mapping_discrete}, we characterize the channel through its CIR, equilibrium gain, and effective channel memory.
We first consider a point release at the TX release state $s_{i^{\mathrm{TX}}}$. We define $\boldsymbol e_i$ as the $i$-th canonical basis vector. The corresponding TX initial-state vector is defined as $\boldsymbol x_{\mathrm{TX}}\triangleq \boldsymbol e_{i^{\mathrm{TX}}}$, representing a single molecule initially released at $s_{i^{\mathrm{TX}}}$. The discrete-time CIR is then
\begin{equation}
    h[n] \triangleq y[n]\big|_{\bm x[0]=\boldsymbol x_{\mathrm{TX}}}
    = \boldsymbol o^{\top}\mathbf{P}^{n}\boldsymbol e_{i^{\mathrm{TX}}},
    \label{eq:cir_def_disc}
\end{equation}
and the corresponding continuous-time surrogate is
\begin{equation}
    h(t) \triangleq y(t)\big|_{\bm x(0)=\boldsymbol x_{\mathrm{TX}}}
    \approx \boldsymbol o^{\top}\exp(\mathbf{A}t)\,\boldsymbol e_{i^{\mathrm{TX}}}.
    \label{eq:cir_def_cts}
\end{equation}

\begin{remark}[Effect of TX Release State]
\label{rem:tx_release_state}
The TX release state is specified by the initial-state vector $\boldsymbol x_{\mathrm{TX}}$ in the CIR definition. If the point release occurs at another free state $s_m$, i.e., $i^{\mathrm{TX}} = m$, the corresponding CIR in~\cref{eq:cir_def_disc} becomes $h_m[n]=\boldsymbol o^{\top}\mathbf P^n\boldsymbol e_m$. Therefore, changing the TX release state $s_{i^{\mathrm{TX}}}$ changes the transient response by changing the relative position of the TX release state with respect to the RX.
\end{remark}

Under the non-saturating probe assumption introduced in~\cref{subsec:physical_and_states}, if the transmitter releases $Q$ independent and identically distributed (i.i.d.) molecules with the same initial distribution, the ensemble-mean observation scales linearly as $Q\,y[n]$. In particular, under a point-release impulse, the ensemble-mean response equals $Q\,h[n]$.
The equilibrium distribution, equivalently the stationary distribution of the Markov chain, is denoted by $\bm x^{\mathrm{eq}}$ and satisfies
\begin{equation}
    \bm{x}^{\mathrm{eq}}=\mathbf{P}\,\bm{x}^{\mathrm{eq}}, \qquad \|\bm{x}^{\mathrm{eq}}\|_1=1.
    \label{eq:xeq_def}
\end{equation}
Under standard ergodicity conditions, the equilibrium distribution is unique and $\bm x[n]\to \bm x^{\mathrm{eq}}$ as $n\to\infty$ for any initial distribution~\cite{levin2017markov}. Accordingly, the equilibrium gain is defined as
\begin{equation}
    h^{\mathrm{eq}} \triangleq \lim_{n\to\infty} h[n] = \boldsymbol o^{\top}\bm{x}^{\mathrm{eq}}.
    \label{eq:heq_def}
\end{equation}

The effective channel memory is governed by the rate at which the state distribution approaches equilibrium. For an ergodic finite-state Markov chain, the dominant transient is associated with the second-largest eigenvalue modulus (SLEM) $|\lambda_1|$ of $\mathbf P$~\cite{levin2017markov}. Following the relaxation-time interpretation in~\cite[Ch.~12.2]{levin2017markov}, we define the characteristic time constant
\begin{equation}
    \tau \triangleq \frac{\Delta t}{1-|\lambda_1|},
    \label{eq:tau_def}
\end{equation}
and the corresponding settling time
\begin{equation}
    t^{\mathrm{eq}} \triangleq 5\,\tau,
    \label{eq:teq_def}
\end{equation}
which provides a practical measure of the effective channel memory.

\begin{proposition}[Exponential convergence with characteristic time]
\label{prop:mixing_rate}
Under the continuous-time surrogate in~\cref{eq:cts_approx}, there exists a constant $C>0$ such that
\begin{equation}
    \|\bm{x}(t)-\bm{x}^{\mathrm{eq}}\|_\infty
    \le C\,\exp\left(-\frac{t}{\tau}\right).
    \label{eq:exp_decay_continuous_tau}
\end{equation}
\end{proposition}
\begin{proof}
See Appendix~\cref{app:proof_prop_mixing}.
\end{proof}

\begin{figure*}[t]
  \centering

  \subfigure[\label{fig:overview:a}Markov-chain microarray model.]{%
    \centering
    \includegraphics[width=0.98\linewidth]{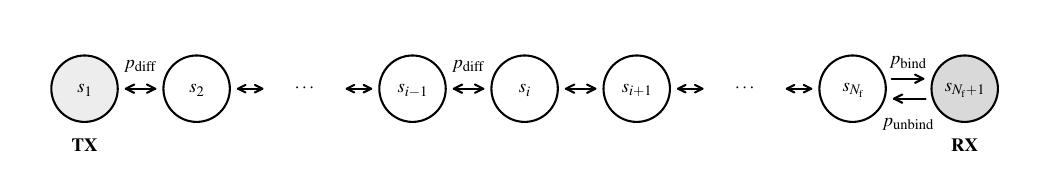}
  }\\[-0.2em]

  \subfigure[\label{fig:overview:b}State evolution heatmap $x_i(t)$.]{%
    \centering
    \includegraphics[width=0.32\linewidth]{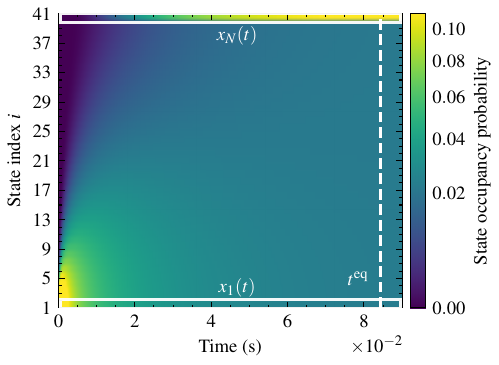}
  }\hfill
  \subfigure[\label{fig:overview:c}Observable response $h(t)=x_{N}(t)$.]{%
    \centering
    \includegraphics[width=0.32\linewidth]{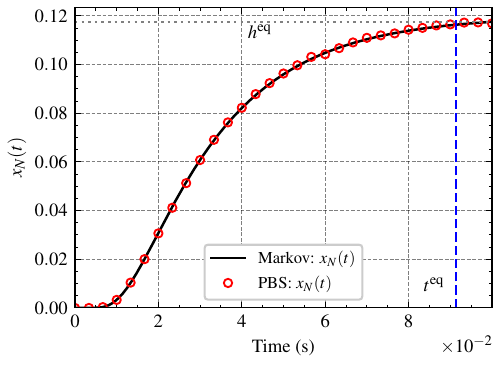}
  }\hfill
  \subfigure[\label{fig:overview:d}Example hidden-state trajectory $x_1(t)$.]{%
    \centering
    \includegraphics[width=0.32\linewidth]{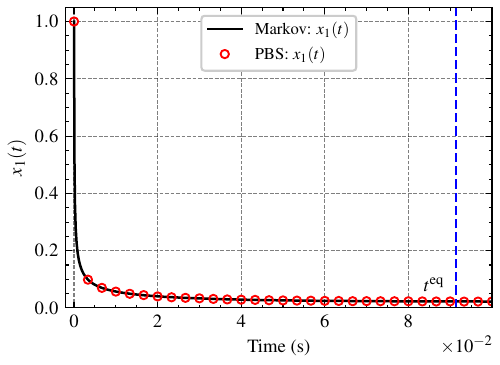}
  }

  \caption{Overview of the proposed Markov-based microarray channel model. (a) A 1D Markov chain with nearest-neighbor diffusion among the free states and reversible binding/unbinding between the receiver-adjacent free state $s_{N_{\mathrm f}}$ and the bound state $s_N$. (b) State occupancy probabilities $x_i(t)$ under a point release at the transmitter, illustrating transient propagation, binding, and convergence toward equilibrium; the highlighted rows correspond to the trajectories shown in (c) and (d). (c) Observable channel response $h(t)=x_N(t)$, showing the equilibrium gain $h^{\mathrm{eq}}$ and settling time $t^{\mathrm{eq}}$. (d) Representative hidden-state response $x_1(t)$, illustrating the evolution of an internal, non-observable state. The solid curves denote the analytical Markov results, while the markers denote PBS results obtained using $Q=5000$ cDNA molecules and averaged over $100$ independent Monte Carlo realizations. Illustrative parameters are selected for visual clarity and may differ from those used in the numerical evaluations.}
  \label{fig:system_overview}
\end{figure*}
\section{Microarray Channel Specialization and Analytical Characterization}
\label{sec:channel_specialization}

In this section, we specialize the general Markov state-space framework in~\cref{sec:general_framework} to the considered DNA microarray RX. The closed-form analytical characterization developed here is specific to an effective 1D microarray channel, which captures the dominant axial transport path from the TX to the reactive RX while averaging out transverse variations. This specialization enables explicit characterization of the CIR, the equilibrium gain, and the effective channel memory.

\subsection{Microarray Channel Specialization}
\label{subsec:microarray_channel_specialization}

We specialize the general construction in~\cref{subsec:transition_matrix} to a 1D microarray geometry with pure diffusion, i.e., $\mathbf v=\mathbf 0$ in~\cref{eq:pde_cd}. The axial channel is discretized into $N_{\mathrm f}$ free states, and the reactive receiver surface is represented by a single bound state. Hence, the state space becomes
\begin{equation}
    \mathcal{S}=\{s_1,\ldots,s_{N_{\mathrm f}},\, s_{N_{\mathrm f}+1}\}.
    \label{eq:state_space_1d}
\end{equation}
For compact notation, we define the total number of states as
\begin{equation}
    N \triangleq N_{\mathrm f}+1,
    \label{eq:N_def_microarray}
\end{equation}
so that the unique RX-associated bound state is $s_N$.

The $N_{\mathrm f}$ free states uniformly discretize the effective one-dimensional diffusion domain into voxels of width $\Delta x$. We define the corresponding effective axial channel length as 
\begin{equation} 
    d_{\mathrm{eff}}\triangleq N_{\mathrm f}\Delta x. 
    \label{eq:d_eff} 
\end{equation} 
For a fixed physical channel, $d_{\mathrm{eff}}$ remains unchanged under spatial-grid refinement; increasing $N_{\mathrm f}$ is accompanied by a corresponding decrease in $\Delta x$. Since the reactive boundary is adjacent to the terminal free state, the receiver-adjacent set reduces to

Since the reactive boundary is adjacent to the terminal free state, the receiver-adjacent set reduces to
\begin{equation}
    \mathcal{V}_{\mathrm{RX}}=\{s_{N_{\mathrm f}}\}.
    \label{eq:VRX_1d}
\end{equation}
In the considered 1D microarray specialization, we set $i^{\mathrm{TX}}=1$, corresponding to a TX point release at the first free state, while other release locations can be handled by changing the initial-state vector $\boldsymbol x_{\mathrm{TX}}$ as discussed in~\cref{rem:tx_release_state}.
Moreover, the observable set consists of the single bound state $s_N$, and thus the observation vector introduced in~\cref{sec:general_framework} reduces to
\begin{equation}
    \boldsymbol o = \boldsymbol e_N,
    \label{eq:o_microarray}
\end{equation}
where $\boldsymbol e_i$ denotes the $i$-th canonical basis vector. Likewise, the TX initial-state vector is
\begin{equation}
    \boldsymbol x_{\mathrm{TX}} \triangleq \boldsymbol e_1.
    \label{eq:tx_vector_microarray}
\end{equation}
Accordingly, the single-molecule observation is $y[k]=\boldsymbol o^{\top}\bm x[k]$.

\Cref{fig:overview:a} illustrates the resulting 1D Markov chain: a molecule performs nearest-neighbor diffusion among the free states $\{s_1,\ldots,s_{N_{\mathrm f}}\}$ with probability $p_{\mathrm{diff}}$ per step, and the terminal free state $s_{N_{\mathrm f}}$ couples to the bound state $s_N$ via binding and unbinding with probabilities $p_{\mathrm{bind}}$ and $p_{\mathrm{unbind}}$, respectively.

With the state ordering $(s_1,\ldots,s_{N_{\mathrm f}}, s_N)$, the transition matrix $\mathbf P$ retains the block form in~\cref{eq:P_blocks} with a single bound state. Using \cref{eq:pd,eq:p_bind_general,eq:p_unbind_general} and $\mathcal V_{\mathrm{RX}}=\{s_{N_{\mathrm f}}\}$, the free-to-free block is
\begin{equation}
    \mathbf{Q}=
    \begin{bmatrix}
    1-p_{\mathrm{diff}}     & p_{\mathrm{diff}}        &        & 0 \\
    p_{\mathrm{diff}}       & 1-2p_{\mathrm{diff}}     & \ddots & \vdots \\
                            & \ddots                   & \ddots & p_{\mathrm{diff}} \\
    0                       & \cdots                   & p_{\mathrm{diff}} & 1-p_{\mathrm{diff}}-p_{\mathrm{bind}}
    \end{bmatrix},
    \label{eq:Q_1d}
\end{equation}
where the last diagonal entry reflects the competition between diffusion and binding at $s_{N_{\mathrm f}}$. Since there is only one bound state, the remaining blocks reduce to
\begin{equation}
    \left\{
    \begin{aligned}
    \mathbf{B} &= [\,0,\ldots,0,\,p_{\mathrm{bind}}\,],\\
    \mathbf{U} &= [\,0,\ldots,0,\,p_{\mathrm{unbind}}\,]^{\top},\\
    \mathbf{R} &= [\,1-p_{\mathrm{unbind}}\,].
    \end{aligned}
    \right.
    \label{eq:BUR_1d}
\end{equation}

Substituting $\boldsymbol o=\boldsymbol e_N$ and $\bm x[0]=\boldsymbol x_{\mathrm{TX}}$ into~\cref{eq:cir_def_disc,eq:cir_def_cts}, the microarray CIR becomes
\begin{equation}
    h[n]=\boldsymbol e_{N}^\top\mathbf P^{n}\boldsymbol e_{1},
    \qquad
    h(t)\approx \boldsymbol e_{N}^\top\exp(\mathbf A t)\,\boldsymbol e_{1},
    \label{eq:cir_1d_explicit}
\end{equation}
where $\mathbf A=(\mathbf P-\mathbf I)/\Delta t$ and $h[n]=h(n\Delta t)$.

\subsection{Equilibrium and Settling Characteristics}
\label{subsec:microarray_eq_settling}

We now characterize the equilibrium behavior and the effective channel memory of the specialized microarray chain. For $0<p_{\mathrm{diff}}<\frac{1}{2}$ and $0<p_{\mathrm{bind}},p_{\mathrm{unbind}}<1$, the chain is irreducible and aperiodic, and is therefore ergodic. This convergence behavior is illustrated by the state-probability heatmap in~\cref{fig:overview:b}, where $\{x_i(t)\}$ evolves from the point-release initial condition and approaches the equilibrium distribution. Hence, there exists a unique equilibrium distribution $\bm x^{\mathrm{eq}}$ satisfying
\begin{equation}
    \bm x^{\mathrm{eq}} = \mathbf P\,\bm x^{\mathrm{eq}}, \qquad
    \|\bm x^{\mathrm{eq}}\|_1=1 .
    \label{eq:xeq_microarray_def}
\end{equation}

\begin{corollary}[Equilibrium Distribution and Gain]
\label{prop:xeq_heq_microarray}
For the 1D microarray specialization in~\cref{subsec:microarray_channel_specialization}, the equilibrium distribution is
\begin{equation}
    x_i^{\mathrm{eq}}=
    \begin{cases}
    \displaystyle\frac{K_{\mathrm D}}{N_{\mathrm f}K_{\mathrm D}+c_{\mathrm p}},
    & 1\le i\le N_{\mathrm f},\\[2mm]
    \displaystyle\frac{c_{\mathrm p}}{N_{\mathrm f}K_{\mathrm D}+c_{\mathrm p}},
    & i=N_{\mathrm f}+1,
    \end{cases}
    \label{eq:xeq_microarray_closedform}
\end{equation}
and the corresponding equilibrium channel gain is
\begin{equation}
    h^{\mathrm{eq}}
    =x_{N_{\mathrm f}+1}^{\mathrm{eq}}
    =\frac{c_{\mathrm p}}{N_{\mathrm f}K_{\mathrm D}+c_{\mathrm p}}
    =\frac{\sigma_{\mathrm p}}{K_{\mathrm D}d_{\mathrm{eff}}+\sigma_{\mathrm p}},
    \label{eq:heq_microarray_closedform}
\end{equation}
where $K_{\mathrm D}\triangleq k_{\mathrm{off}}/k_{\mathrm{on}}$ is the dissociation constant, and the last equality follows from~\cref{eq:cp,eq:d_eff}.
\end{corollary}

\begin{proof}
For the homogeneous 1D diffusion chain, the free-state equilibrium probabilities are identical. Together with the interface balance
\[
x_{N_{\mathrm f}}^{\mathrm{eq}}p_{\mathrm{bind}}=x_{N_{\mathrm f}+1}^{\mathrm{eq}}p_{\mathrm{unbind}}
\]
and the normalization condition $\sum_{i=1}^{N}x_i^{\mathrm{eq}}=1$, this yields~\cref{eq:xeq_microarray_closedform,eq:heq_microarray_closedform}.
\end{proof}

The expression in~\cref{eq:heq_microarray_closedform} shows that the equilibrium gain depends on the physical surface probe density $\sigma_{\mathrm p}$, the dissociation constant $K_{\mathrm D}$, and the effective axial channel length $d_{\mathrm{eff}}=N_{\mathrm f}\Delta x$, rather than on $N_{\mathrm f}$ alone, and therefore remains unchanged under spatial-grid refinement for fixed physical channel and reaction parameters. Since the observable state is $s_N$, the continuous-time surrogate response satisfies $h(t)=x_N(t)$, and \cref{fig:overview:c} illustrates how $h(t)$ approaches the equilibrium gain $h^{\mathrm{eq}}$ under a point release. This dependence is further illustrated in~\cref{fig:eq_settling:a}: varying $K_{\mathrm D}$ changes the equilibrium gain, while the settling behavior remains relatively similar when $k_{\mathrm{off}}$ is fixed.

With the characteristic time $\tau$ and the settling time $t^{\mathrm{eq}}$ defined in~\cref{eq:tau_def,eq:teq_def}, the effective channel memory is governed by the slowest transient mode of the Markov chain. In particular, a larger SLEM $|\lambda_1|$ implies slower relaxation, a larger settling time, and hence a longer effective channel memory. This trend is illustrated in~\cref{fig:eq_settling:b}: decreasing $k_{\mathrm{off}}$ slows the convergence toward equilibrium and increases $t^{\mathrm{eq}}$, while the equilibrium gain remains unchanged when $K_{\mathrm D}$ is fixed. A similar relaxation behavior is also observed in hidden free-state trajectories, as illustrated by $x_1(t)$ in~\cref{fig:overview:d}.

Overall, the analytical state responses in~\cref{fig:overview:c,fig:overview:d} closely agree with the PBS results, validating the transient and equilibrium behavior predicted by the proposed Markov model. In addition, the gradual relaxation of the bound population toward equilibrium is qualitatively consistent with the DNA affinity-biosensor response reported in~\cite{hassibi2005biological}.

\begin{remark}[Design implication]
\label{rem:design_Tb_teq}
The ratio $T_\mathrm{b}/t^{\mathrm{eq}}$ provides a convenient indicator of ISI severity. When $T_\mathrm{b} \gtrsim t^{\mathrm{eq}}$, the channel is effectively weak-memory, and the residual ISI is small. In contrast, when $T_\mathrm{b} \lesssim t^{\mathrm{eq}}$, the channel exhibits strong memory, and history-dependent detection becomes beneficial.
\end{remark}

\begin{figure*}[t]
    \centering
    \subfigure[\label{fig:eq_settling:a}Impact of $K_{\mathrm D}$ on the equilibrium gain $h^{\mathrm{eq}}$ ($k_{\mathrm{off}}$ fixed).]{%
        \centering
        \includegraphics[width=0.45\linewidth]{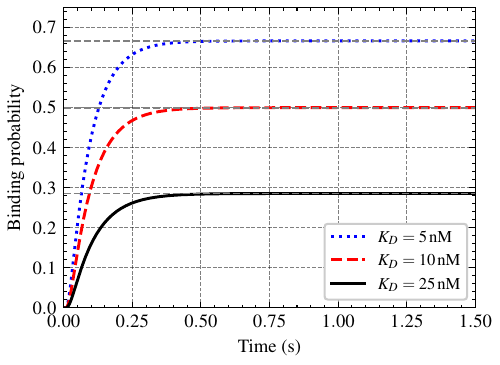}}
    \hfill
    \subfigure[\label{fig:eq_settling:b}Impact of $k_{\mathrm{off}}$ on the settling behavior ($K_{\mathrm D}$ fixed).]{%
        \centering
        \includegraphics[width=0.45\linewidth]{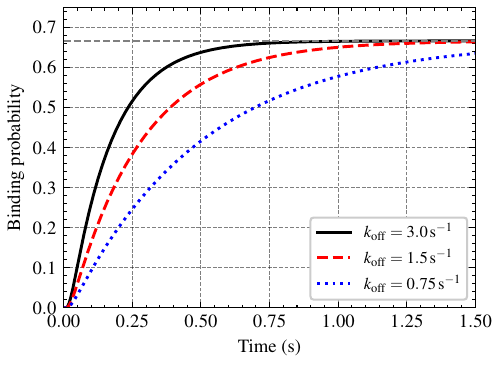}}
    \caption{Bound-state response $h(t)$ under the point-release initial condition $\bm x[0]=\boldsymbol e_1$. (a) Varying the dissociation constant $K_{\mathrm D}$ changes the equilibrium gain $h^{\mathrm{eq}}$ predicted by~\cref{eq:heq_microarray_closedform}, while the settling time $t^{\mathrm{eq}}$ remains relatively similar when $k_{\mathrm{off}}$ is fixed. (b) Varying the unbinding rate constnt $k_{\mathrm{off}}$ mainly changes the settling time $t^{\mathrm{eq}}$, and hence the effective channel memory, while all curves share the same equilibrium gain $h^{\mathrm{eq}}$ when $K_{\mathrm D}$ is fixed.}
    \label{fig:eq_settling}
\end{figure*}

\subsection{Applicability and Limitations of the 1D Specialization}

The considered 1D specialization is suitable for diffusion-dominated microarray-style reception with approximately parallel and laterally uniform transmitter and receiver surfaces, where the dominant concentration gradient is normal to the reactive receiver. Similar 1D or effectively 1D diffusion-reaction models have been adopted for DNA microarrays~\cite{pappaert2003diffusion} and surface-based biosensors~\cite{hassibi2005biological,squires2008making}. In this setting, the single bound state represents a homogeneous probe population with common binding and unbinding kinetics, and the derived CIR, covariance, and BER-related expressions provide baseline analytical predictions.

For finite receiver regions, partial receiver coverage, nonuniform probe distributions, spatially varying hybridization kinetics, transverse diffusion, edge effects, appreciable flow, probe-site saturation, or higher-dimensional geometries, the response may deviate from the 1D analytical results. These effects can be incorporated into the proposed Markov framework by extending the free and bound state spaces or modifying the corresponding blocks of $\mathbf P$ and the observation vector $\boldsymbol o$. The 2D example in~\cref{sec:supp_2d} of the Supplementary Material illustrates this extensibility, while systematic approximation-error analysis is left for future work.

\section{Symbol-Rate Communication Model and Noise Statistics}
\label{sec:io_noise}

In this section, we extend the single-molecule Markov channel in~\cref{sec:channel_specialization} to a symbol-rate communication model for detection. We first introduce on-off keying (OOK) and derive the symbol-rate mean response in convolutional form. We then characterize the associated \emph{counting noise}, with particular emphasis on its symbol-rate covariance. Hence, this section captures both the mean-domain memory of the channel through ISI and the noise-domain memory through temporally correlated counting noise.

\subsection{Transmitter Modulation and Signaling}
\label{subsec:tx_modulation}

We employ concentration-based binary OOK. Let $Q\in\mathbb{N}$ denote the number of cDNA molecules released by the TX per binary-1 symbol. At the beginning of symbol interval $k$, i.e., at fine-time index $n=kN_s$, the TX releases $Q$ molecules when $a_k=1$ and releases no molecules when $a_k=0$~\cite{zheng2020noise,zheng2025anis}. The symbol interval is aligned with the Markov sampling grid through
\begin{equation}
    T_\mathrm{b} \triangleq N_s\Delta t, \qquad N_s\in\mathbb N.
    \label{eq:Tb_alignment}
\end{equation}
At the $\Delta t$ resolution, the release sequence is represented by the impulse train
\begin{equation}
    u[n] \triangleq \sum_{k=0}^{\infty} a_k\,\delta[n-kN_s],
    \label{eq:uk_def}
\end{equation}
where $\delta[\cdot]$ denotes the Kronecker delta.

A point release produces the single-molecule CIR $h[n]$ on the $\Delta t$ grid; see~\cref{eq:cir_1d_explicit}. Under the non-saturating probe assumption introduced in~\cref{sec:general_framework}, the molecules within each release evolve independently according to the same single-molecule dynamics. By linear superposition, the conditional mean response at fine-time index $n$ is
\begin{equation}
    \bar z[n]
    = Q\,(u*h)[n]
    = Q\sum_{r=0}^{\lfloor n/N_s\rfloor} a_r\,h[n-rN_s],
    \label{eq:zk_bar_convolution}
\end{equation}
where $*$ denotes convolution.

Sampling~\cref{eq:zk_bar_convolution} at symbol boundaries yields the symbol-rate mean sequence
\begin{equation}
    \bar z_k \triangleq \bar z[kN_s], \qquad k=0,1,2,\ldots .
    \label{eq:zk_bar_symbol_def}
\end{equation}
Likewise, we define the symbol-spaced CIR taps as
\begin{equation}
    h_\ell \triangleq h[\ell N_s], \qquad \ell=0,1,2,\ldots .
    \label{eq:h_symbol_taps}
\end{equation}
Substituting~\cref{eq:h_symbol_taps} into the sampled version of~\cref{eq:zk_bar_convolution} gives the standard symbol-rate ISI model
\begin{equation}
    \bar z_k
    = Q\sum_{\ell=0}^{k} a_{k-\ell}\,h_{\ell}.
    \label{eq:zk_bar_symbol_isi}
\end{equation}
The symbol-rate observation therefore depends on the current symbol as well as past releases through the CIR taps $\{h_\ell\}$.

\subsection{Counting Noise Model and Statistics}
\label{subsec:noise_statistics}

The RX observation is the \emph{number of bound cDNA molecules} on the microarray surface. Let $z[n]\in\mathbb{Z}_{\ge 0}$ denote the fine-time bound-molecule count at time $t=n\Delta t$. Sampling at symbol boundaries gives the \emph{symbol-rate observation}
\begin{equation}
    z_k \triangleq z[kN_s].
    \label{eq:zk_def_noisy}
\end{equation}
Unless stated otherwise, all expectations, variances, and covariances in this section are conditioned on the transmit sequence $\{a_k\}$.

Under the non-saturating probe assumption, molecules evolve independently, and the total observation can be represented as the superposition of independent batch contributions. Consider the release triggered at symbol time $r$, which contains $Qa_r$ molecules. For any fine-time index $n\ge rN_s$, each molecule from this batch is bound with probability $h[n-rN_s]$. Hence,
\begin{equation}
    z^{(r)}[n]\sim \mathcal{B}\!\left(Qa_r,\; h[n-rN_s]\right),\qquad n\ge rN_s,
    \label{eq:binom_batch_dt}
\end{equation}
where $z^{(r)}[n]=0$ for $n<rN_s$. The total fine-time observation is then
\begin{equation}
    z[n]\triangleq \sum_{r=0}^{\lfloor n/N_s\rfloor} z^{(r)}[n].
    \label{eq:z_superposition_dt}
\end{equation}
Accordingly, for a fixed $n$, $z[n]$ follows a \textit{grouped-binomial model}, or equivalently a Poisson-binomial distribution.

Taking expectation in~\cref{eq:z_superposition_dt} yields
\begin{equation}
    \E\!\left\{z[n]\right\}
    = \sum_{r=0}^{\lfloor n/N_s\rfloor} Qa_r\,h[n-rN_s]
    \triangleq \bar z[n],
    \label{eq:mean_z_dt}
\end{equation}
which is consistent with~\cref{eq:zk_bar_convolution}. Since different batches correspond to disjoint molecule sets and are therefore independent, the variance is
\begin{equation}
\begin{aligned}
    \Var\!\left(z[n]\right)
    &= \sum_{r=0}^{\lfloor n/N_s\rfloor} Qa_r\,h[n-rN_s]\bigl(1-h[n-rN_s]\bigr) \\
    &\triangleq \sigma_z^2[n].
\end{aligned}
\label{eq:var_z_dt}
\end{equation}

We define the \emph{counting noise} at fine-time index $n$ as
\begin{equation}
    w[n]\triangleq z[n]-\bar z[n].
    \label{eq:w_dt_def}
\end{equation}
By construction, $\E\{w[n]\}=0$ and $\Var(w[n])=\sigma_z^2[n]$.
At the symbol rate, we define the counting noise as
\begin{equation}
    w_k\triangleq w[kN_s],
    \label{eq:wk_def}
\end{equation}
Together with the symbol-rate observation $z_k$ defined in~\cref{eq:zk_def_noisy}, this gives
\begin{equation}
    z_k=\bar z_k+w_k,
    \label{eq:zk_decomp}
\end{equation}
where $\bar z_k$ is the conditional mean given by~\cref{eq:zk_bar_symbol_isi}. Sampling~\cref{eq:var_z_dt} at symbol boundaries gives
\begin{equation}
    \sigma_k^2\triangleq\Var(z_k)=\Var(w_k)=\sum_{\ell=0}^{k}Qa_{k-\ell}\,h_\ell\bigl(1-h_\ell\bigr).
    \label{eq:var_z_symbol}
\end{equation}

The grouped-binomial observation results from the superposition of independent batch-wise molecular contributions. When the aggregate conditional mean $\bar z_k$ and variance $\sigma_k^2$ are sufficiently large, the counting noise can be marginally approximated as Gaussian:
\begin{equation}
    w_k\sim\mathcal{N}\!\left(0,\sigma_k^2\right).
    \label{eq:gaussian_approximation}
\end{equation}
The accuracy of this approximation is validated in Appendix~\cref{app:gaussian_validation} under both low- and high-count conditions. This approximation is adopted only for receiver-level analysis; the grouped-binomial formulation and the covariance result derived below do not rely on Gaussianity.

Due to reversible binding and unbinding, the sequence $\{w_k\}$ is generally \emph{colored}, i.e.,
\[
\Cov(w_k,w_{k-\ell})\neq 0
\]
for some $\ell\ge 1$. Since $\bar z_k$ is deterministic conditioned on $\{a_k\}$, centering does not change covariance, and hence
\[
\Cov(w_k,w_{k-\ell})=\Cov(z_k,z_{k-\ell}).
\]
Therefore, the symbol-rate covariance characterizes the temporal correlation structure of the counting noise.

\begin{proposition}[Covariance of the counting noise]
\label{prop:cov_microarray}
Under the non-saturating probe assumption, consider the microarray RX with a single observable bound state $s_N$, i.e., $\boldsymbol o=\boldsymbol e_N$. Define the $\ell$-symbol return probability
\begin{equation}
    \pi_N^{(\ell)} \triangleq (\mathbf P^{\ell N_s})_{N,N}
    = \boldsymbol e_N^{\top}\mathbf P^{\ell N_s}\boldsymbol e_N,
    \qquad \ell=1,2,\ldots,
    \label{eq:pi_N_l_def}
\end{equation}
which is the probability that a molecule in $s_N$ at one symbol boundary is again in $s_N$ after $\ell$ symbol intervals. Then, conditioned on the transmit sequence $\{a_k\}$, for $k\ge \ell\ge 1$, the lag-$\ell$ covariance of the symbol-rate counting noise satisfies
\begin{equation}
    \Cov(w_k,w_{k-\ell})
    = Q\sum_{r=0}^{k-\ell} a_r\,h_{k-\ell-r}\Bigl(\pi_N^{(\ell)}-h_{k-r}\Bigr).
    \label{eq:cov_wk_wk_l_microarray}
\end{equation}
\end{proposition}

\begin{proof}
See Appendix~\cref{app:proof_cov_microarray}.
\end{proof}

\begin{remark}[Higher-Dimensional Covariance Analysis]
\label{rem:hd_covariance}
Although the covariance expression in~\cref{prop:cov_microarray} is derived for the 1D specialization, the same Markov state-space procedure applies to 2D/3D geometries by replacing $\mathbf P$ and $\boldsymbol o$ with their higher-dimensional counterparts. The resulting covariance generally needs to be evaluated numerically because transverse diffusion, edge effects, and nonuniform binding are encoded in $\mathbf P$.
\end{remark}

To characterize the colored nature of the counting noise in a dimensionless form, we define the lag-$\ell$ correlation coefficient
\begin{equation}
    \rho_k[\ell]\triangleq
    \frac{\Cov(w_k,w_{k-\ell})}{\sigma_k\,\sigma_{k-\ell}},
    \qquad k\ge \ell\ge 1.
    \label{eq:corrcoef_w}
\end{equation}

The key physical reason for this temporal dependence is that the same molecule may contribute to multiple symbol-rate observations. Unlike an absorbing RX, a molecule that is bound at time $k-\ell$ may still be bound at time $k$, or may unbind and later rebind before time $k$. Hence, the symbol-rate observation exhibits not only mean-domain memory through the ISI taps $\{h_\ell\}$, but also noise-domain memory through the covariance of the counting noise. The return probability $\pi_N^{(\ell)}$ compactly captures this persistence across symbol intervals.

At symbol time $k$, the observation $z_k$ counts how many molecules are bound on the microarray surface. Unlike an absorbing receiver, a molecule that is bound at time $k-\ell$ may still be bound at time $k$ (or may unbind and rebind before time $k$). Hence, the same molecule can contribute to multiple samples, creating statistical dependence over time and yielding colored counting noise. The return probability $\pi_N^{(\ell)}=(\mathbf P^{\ell N_s})_{N,N}$ compactly captures this persistence across $\ell$ symbol intervals.

\begin{corollary}[Lag-dependent covariance decay]
\label{cor:lag_decorrelation}
For the ergodic microarray Markov chain, under the continuous-time surrogate in~\cref{eq:cts_approx}, there exists a constant $C>0$ such that, for all valid indices $k\ge \ell\ge 1$,
\begin{equation}
    \bigl|\Cov(w_k,w_{k-\ell})\bigr|
    \le C\,(k-\ell+1)\exp\!\left(-\frac{\ell T_\mathrm{b}}{\tau}\right),
    \label{eq:cov_decay_bound_tau}
\end{equation}
where $\tau$ is the characteristic time defined in~\cref{eq:tau_def}. Hence, the temporal dependence of the counting noise decays with the lag $\ell$.
\end{corollary}

\begin{proof}
See Appendix~\cref{app:proof_cor_lag_decorrelation_tau}.
\end{proof}

\begin{figure}[t]
    \centering
    \subfigure[$T_\mathrm{b}=0.02$ s\label{fig:rho-a}]{\includegraphics[width=0.493\linewidth]{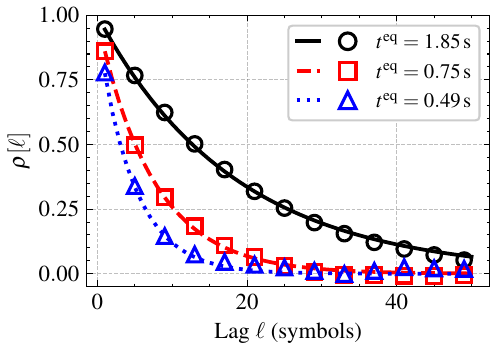}}\hfill
    \subfigure[$T_\mathrm{b}=0.05$ s\label{fig:rho-b}]{\includegraphics[width=0.493\linewidth]{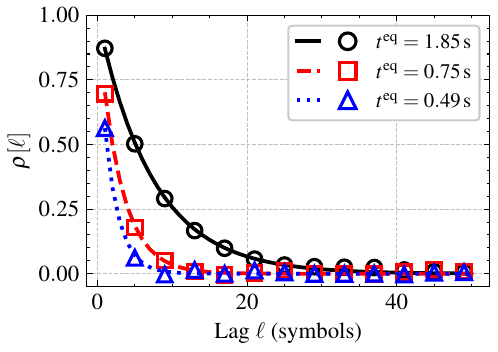}}\\
    \subfigure[$T_\mathrm{b}=0.1$ s\label{fig:rho-c}]{\includegraphics[width=0.493\linewidth]{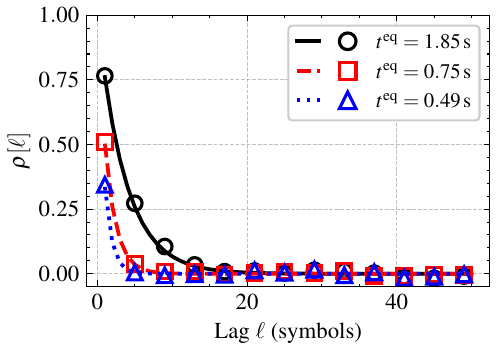}}
    \subfigure[$T_\mathrm{b}=0.3$ s\label{fig:rho-d}]{\includegraphics[width=0.493\linewidth]{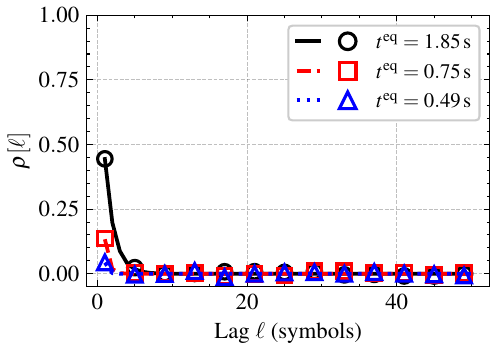}}\hfill
    \caption{Lag-dependent correlation coefficient $\rho[\ell]$ of the counting noise for independent equiprobable OOK signaling. For each lag $\ell$, the analytical curves are obtained from the theoretical covariance and variances in~\cref{eq:corrcoef_w}, whereas the markers estimate the same normalized correlation from simulated counting-noise samples; both are evaluated over valid post-warm-up symbol times and independent OOK realizations. Different subfigures correspond to different symbol intervals $T_\mathrm{b}$, and different curves correspond to different settling times $t^{\mathrm{eq}}$ with fixed equilibrium gain $h^{\mathrm{eq}}$.}
    \label{fig:rho_vs_ell}
\end{figure}

\Cref{fig:rho_vs_ell} compares the analytical and Monte Carlo results of the lag-dependent correlation coefficient $\rho[\ell]$ of the counting noise for different symbol intervals $T_\mathrm{b}\in\{0.02,0.05,0.1,0.3\}$~s and settling times $t^{\mathrm{eq}}\in\{1.85,0.75,0.49\}$~s, while keeping the equilibrium gain $h^{\mathrm{eq}}$ fixed. The solid lines are calculated from the theoretical expression in~\cref{eq:corrcoef_w}, whereas the markers are Monte Carlo estimates obtained by generating random OOK sequences, propagating the corresponding stochastic symbol-rate molecule-count trajectories, and estimating $\rho[\ell]$ from the resulting counting-noise samples. Therefore, the markers are not obtained by directly evaluating the closed-form correlation expression, but provide an empirical validation of the colored counting-noise covariance analysis within the stochastic Markov framework.
\footnote{A direct comparison of the unnormalized covariance predicted by Proposition~2 with Monte Carlo estimates is provided in~\cref{sec:supp_cov_validation} of the Supplementary Material.}
A common trend in all subfigures is that $\rho[\ell]$ decreases monotonically with the lag $\ell$ and eventually approaches zero, which is consistent with the exponential covariance decay predicted by~\cref{cor:lag_decorrelation}.

For a fixed symbol interval $T_\mathrm{b}$, a larger settling time produces a noticeably longer correlation tail: in each subfigure, the curve for $t^{\mathrm{eq}}=1.85$~s lies above that for $t^{\mathrm{eq}}=0.75$~s, which in turn lies above that for $t^{\mathrm{eq}}=0.49$~s. This behavior is consistent with~\cref{eq:cov_wk_wk_l_microarray}, since a larger $t^{\mathrm{eq}}$ implies slower relaxation of the underlying Markov chain and hence a slower decay of the bound-state return probability $\pi_N^{(\ell)}=(\mathbf P^{\ell N_s})_{N,N}$. As a result, the covariance term $\Cov(w_k,w_{k-\ell})$ remains significant over a wider range of lags, yielding stronger temporal correlation in the counting noise. This is also consistent with the fact that a larger $t^{\mathrm{eq}}$ corresponds to a longer effective channel memory and stronger ISI.

For a fixed settling time, increasing the symbol interval $T_\mathrm{b}$ accelerates the decay of $\rho[\ell]$ in the symbol domain. This trend is evident when comparing \cref{fig:rho-a,fig:rho-b,fig:rho-c,fig:rho-d}: as $T_\mathrm{b}$ increases from $0.02$~s to $0.3$~s, the correlation tail becomes progressively shorter. The reason is that a larger $T_\mathrm{b}$ means that the same lag $\ell$ corresponds to a larger physical-time separation $\ell T_\mathrm{b}$, so the system has more time to relax toward equilibrium between two symbol-rate observations. Consequently, the counting noise becomes less temporally correlated.

Overall, the correlation span in the symbol domain is governed by the ratio $t^{\mathrm{eq}}/T_\mathrm{b}$. When $T_\mathrm{b}\ll t^{\mathrm{eq}}$, the counting noise remains correlated over many symbols; when $T_\mathrm{b}\gg t^{\mathrm{eq}}$, the temporal correlation becomes weak and symbol-wise memoryless receiver models become increasingly accurate.

\section{Differential Observation and Detection}
\label{sec:detection}

In this section, we develop detection methods for the symbol-rate observation sequence $\{z_k\}$. Because microarray-style reception is inventory-like under reversible hybridization, the symbol-rate observation contains a nonvanishing baseline, while the channel exhibits long effective memory, making direct level detection unreliable in strong-memory regimes. To mitigate this effect, we first transform the observation sequence into the differential domain and characterize the resulting differential-domain ISI. Based on this formulation, we develop two low-complexity receivers: a memoryless differential-threshold detector and a finite-memory DFE. In addition, MAP-DF is included as a stronger statistics-aware benchmark for evaluating the proposed low-complexity receivers.

\begin{figure}[t]
    \centering
    \subfigure[$T_\mathrm{b}=0.02$ s\label{fig:delta_h-a}]{\includegraphics[width=0.493\linewidth]{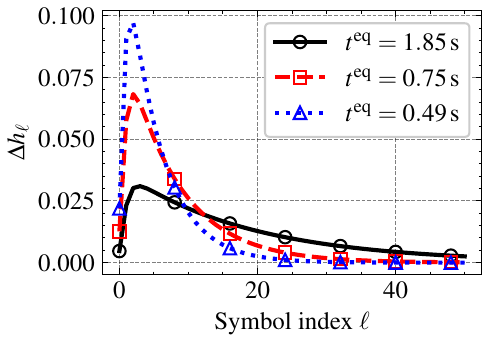}}\hfill
    \subfigure[$T_\mathrm{b}=0.05$ s\label{fig:delta_h-b}]{\includegraphics[width=0.493\linewidth]{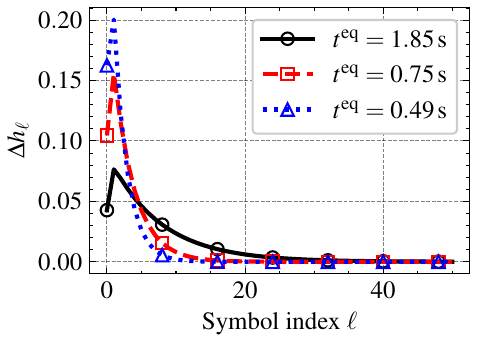}}\\
    \subfigure[$T_\mathrm{b}=0.1$ s\label{fig:delta_h-c}]{\includegraphics[width=0.493\linewidth]{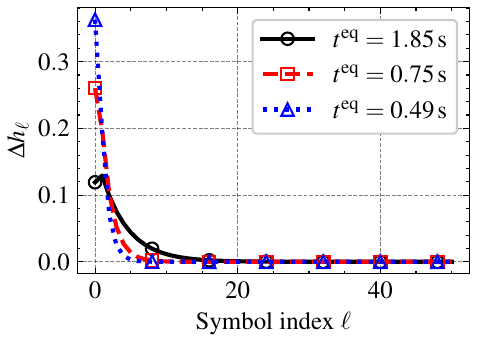}}
    \subfigure[$T_\mathrm{b}=0.3$ s\label{fig:delta_h-d}]{\includegraphics[width=0.493\linewidth]{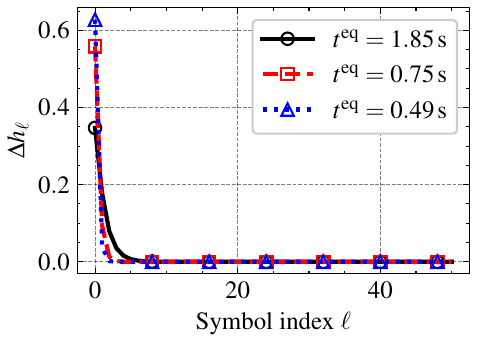}}\hfill
    \caption{Differential CIR taps $\Delta h_\ell$ for symbol intervals $T_\mathrm{b}\in\{0.02,\,0.05,\,0.1,\,0.3\}$~s and settling times $t^{\mathrm{eq}}\in\{1.85,\,0.75,\,0.49\}$~s, with fixed equilibrium gain $h^{\mathrm{eq}}$. Larger $t^{\mathrm{eq}}$ yields a post-cursor ISI tail, while larger $T_\mathrm{b}$ reduces the ISI span in the symbol domain.}
    \label{fig:cir_taps_tb_teq}
\end{figure}

\subsection{Differential Observation and ISI Characterization}
\label{subsec:det_delta_isi}

The symbol-rate observation $z_k$ represents the number of bound molecules at symbol boundaries and is therefore inventory-like under reversible hybridization: once bound, molecules may remain on the receiver surface or return to it after unbinding. As a result, the observation sequence contains a slowly varying baseline and long temporal memory, which can mask the symbol-dependent increment and degrade direct level detection in strong-memory regimes. To suppress this baseline while preserving the symbol-induced change, we consider the first-order difference
\begin{equation}
    \Delta z_k \triangleq z_{k+1}-z_k,\qquad k=0,1,2,\ldots,
    \label{eq:delta_z_def}
\end{equation}
which represents the net change of the bound-state population over one symbol interval.

Using the symbol-rate observation model in~\cref{eq:zk_decomp}, the differential observation in~\cref{eq:delta_z_def} can be written as
\begin{equation}
    \Delta z_k = (\bar z_{k+1}-\bar z_k) + \Delta w_k,
    \label{eq:delta_z_decomp}
\end{equation}
where
\begin{equation}
    \Delta w_k \triangleq w_{k+1}-w_k
\end{equation}
denotes the differenced counting noise. Since $w_k$ is zero-mean, $\Delta w_k$ is also zero-mean, with variance
\begin{equation}
    \Var(\Delta w_k)=\sigma_{k+1}^2+\sigma_k^2-2\Cov(w_{k+1},w_k),
    \label{eq:var_delta_w}
\end{equation}
where $\sigma_k^2$ is given in~\cref{eq:var_z_symbol} and $\Cov(w_{k+1},w_k)$ follows from~\cref{prop:cov_microarray} with $\ell=1$.

Substituting~\cref{eq:zk_bar_symbol_isi} into~\cref{eq:delta_z_decomp} yields the differential-domain ISI model
\begin{equation}
    \Delta z_k
    = Q\sum_{\ell=0}^{k} a_{k-\ell}\,\Delta h_\ell + \Delta w_k,
    \label{eq:delta_z_isi}
\end{equation}
where the differential CIR taps are defined as
\begin{equation}
    \Delta h_\ell \triangleq h_{\ell+1}-h_\ell,\qquad \ell=0,1,2,\ldots .
    \label{eq:delta_h_def}
\end{equation}
Since $h_\ell\to h^{\mathrm{eq}}$ as $\ell$ increases, it follows that $\Delta h_\ell\to 0$, i.e., the differential CIR taps decay as the channel approaches equilibrium.

Separating the current-symbol contribution in~\cref{eq:delta_z_isi} gives
\begin{equation}
    \Delta z_k
    = \underbrace{Q a_k\,\Delta h_0}_{\text{current-symbol term}}
    + \underbrace{Q\sum_{\ell=1}^{k}a_{k-\ell}\,\Delta h_\ell}_{\text{post-cursor ISI term}}
    + \Delta w_k ,
    \label{eq:delta_z_signal_isi}
\end{equation}
where the first term carries the current symbol information, while the remaining taps $\{\Delta h_\ell\}_{\ell\ge 1}$ quantify the post-cursor ISI in the differential domain.

\Cref{fig:cir_taps_tb_teq} illustrates the differential CIR taps for different symbol intervals and settling times. A larger settling time produces a longer tail of $\Delta h_\ell$, which is consistent with a longer effective channel memory. In contrast, increasing $T_\mathrm{b}$ compresses the physical-time memory into fewer symbol intervals and therefore shortens the ISI span in the symbol domain. Thus, the structure of $\{\Delta h_\ell\}$ provides a direct indication of whether a simple differential-threshold detector is sufficient or whether finite-memory feedback-based cancellation is needed.

Specifically, \cref{fig:delta_h-a,fig:delta_h-b} correspond to strong-memory regimes with small symbol intervals, where several post-cursor taps remain significant and decay slowly. \Cref{fig:delta_h-c} represents an intermediate regime, in which only the first few post-cursor taps remain relevant. Finally, \cref{fig:delta_h-d} corresponds to a weak-memory regime, where the post-cursor taps become negligible after only a few lags and the post-cursor ISI term in~\cref{eq:delta_z_signal_isi} can be largely ignored.

\subsection{Low-Complexity Differential Detection}
\label{subsec:det_diff}

We now develop practical detectors for the differential observation sequence $\{\Delta z_k\}$.
In weak-memory regimes, the post-cursor ISI term in~\cref{eq:delta_z_signal_isi} is small, and the differential observation is dominated by the current-symbol component $Q a_k\Delta h_0$. This motivates a memoryless \emph{differential-threshold detector}:
\begin{equation}
    \hat a_k=
    \begin{cases}
    1, & \Delta z_k>\eta,\\
    0, & \text{otherwise}.
    \end{cases}
    \label{eq:det_diff_thresh}
\end{equation}
Under the weak-memory approximation, the conditional means are approximately
\begin{subequations}
\label{eq:mu_weak_isi}
\begin{align}
    \mu_{0,k}
    &\triangleq
    \E\!\left\{\Delta z_k \mid a_k=0\right\}
    \approx 0,
    \label{eq:mu0_weak_isi}
    \\
    \mu_{1,k}
    &\triangleq
    \E\!\left\{\Delta z_k \mid a_k=1\right\}
    \approx Q\,\Delta h_0.
    \label{eq:mu1_weak_isi}
\end{align}
\end{subequations}
Assuming equiprobable OOK, i.e., $\Pr\{a_k=1\}=1/2$, we adopt the midpoint threshold
\begin{equation}
    \eta \approx \frac{\mu_{0,k}+\mu_{1,k}}{2}
    \approx \frac{Q\,\Delta h_0}{2}.
    \label{eq:eta_diff}
\end{equation}
For the considered microarray specialization, $h_0=0$ because the transmitter state is a free state while the observable state is the bound state. Hence, $\Delta h_0=h_1$ in this case.

When the effective channel memory is long, the post-cursor ISI term in~\cref{eq:delta_z_signal_isi} can no longer be ignored, and the memoryless differential-threshold detector becomes inadequate. We therefore consider a finite-memory DFE, which estimates and cancels the dominant post-cursor ISI using previous decisions. For a chosen feedback-memory length $L$, the estimated post-cursor ISI is
\begin{equation}
    \widehat{\mathrm{ISI}}_k
    \triangleq Q\sum_{\ell=1}^{L}\hat a_{k-\ell}\,\Delta h_\ell,
    \label{eq:isi_hat_delta}
\end{equation}
where $\hat a_n=0$ for $n<0$ during receiver initialization. Subtracting this estimate from the differential observation yields the residual decision statistic
\begin{equation}
    r_k\triangleq\Delta z_k-\widehat{\mathrm{ISI}}_k.
    \label{eq:rk_def}
\end{equation}
The DFE decision rule is
\begin{equation}
    \hat a_k=
    \begin{cases}
    1, & r_k>\eta^{\mathrm{DF}},\\
    0, & \text{otherwise},
    \end{cases}
    \label{eq:det_dfe}
\end{equation}
where $\eta^{\mathrm{DF}}$ denotes the DFE threshold.

In this work, to isolate the gain due to ISI cancellation rather than threshold optimization, we use the same midpoint threshold as in~\cref{eq:eta_diff}, i.e.,
\begin{equation}
    \eta^{\mathrm{DF}}=\eta=\frac{Q\,\Delta h_0}{2}.
\end{equation}
Accordingly, the resulting performance improvement primarily reflects the benefit of decision-feedback ISI cancellation.

\subsection{MAP-DF Detection}
\label{subsec:map_df}

We additionally include MAP-DF as a stronger statistics-aware benchmark for the low-complexity receivers developed in~\cref{subsec:det_diff}. At each symbol time, MAP-DF evaluates the conditional likelihood of each candidate symbol using all available previous hard decisions and the corresponding conditional mean and variance of the differential observation. It therefore avoids a fixed feedback-memory truncation while retaining a symbol-by-symbol decision structure.

For each candidate symbol $b\in\{0,1\}$, the conditional differential observation is approximated as
\begin{equation}
    \Delta z_k \mid a_k=b,\hat{\boldsymbol a}_{k-1}\sim\mathcal{N}\left(\mu_k^{(b)},v_k^{(b)}\right),
    \label{eq:map_df_distribution}
\end{equation}
where $\hat{\boldsymbol a}_{k-1}\triangleq[\hat a_0,\hat a_1,\ldots,\hat a_{k-1}]^\top$ collects the previous hard decisions available at symbol time $k$. From~\cref{eq:delta_z_isi}, the conditional mean is reconstructed as
\begin{equation}
    \mu_k^{(b)}=Qb\,\Delta h_0+Q\sum_{\ell=1}^{k}\hat a_{k-\ell}\Delta h_\ell.
    \label{eq:map_df_mean}
\end{equation}
The corresponding conditional variance $v_k^{(b)}$ is evaluated from~\cref{eq:var_delta_w} by setting $a_k=b$ and replacing the previous transmitted symbols with their hard decisions.

Since equiprobable OOK symbols are assumed, the prior probabilities are identical and can be omitted. The MAP-DF decision rule therefore reduces to
\begin{equation}
    \hat a_k=\underset{b\in\{0,1\}}{\arg\max}\left[-\frac{1}{2}\log v_k^{(b)}-\frac{\left(\Delta z_k-\mu_k^{(b)}\right)^2}{2v_k^{(b)}}\right].
    \label{eq:map_df_rule}
\end{equation}
Overall, MAP-DF serves as a stronger symbol-by-symbol benchmark for quantifying the performance gap relative to the proposed low-complexity receivers.


\section{Numerical Results}
\label{sec:numerical_results}

In this section, we evaluate the proposed low-complexity RXs, namely the differential-threshold detector in~\cref{eq:det_diff_thresh} and the DFE in~\cref{eq:det_dfe} with feedback-memory lengths $L\in\{1,3,5\}$. We also include MAP-DF in~\cref{eq:map_df_rule} as a stronger statistics-aware decision-feedback benchmark. All results are obtained using the symbol-rate model developed in~\cref{sec:io_noise,sec:detection}, where the channel response is computed from~\cref{eq:cir_1d_explicit} and the differential observation is given by~\cref{eq:delta_z_isi}. For direct validation, a small-scale comparison between PBS and the analytical model is provided in~\cref{sec:c4:validation} of the Supplementary Material. For MAP-DF, all available previous hard decisions and the corresponding post-cursor taps are used at each symbol time. In contrast, the DFE uses only the first $L$ post-cursor taps. Unless stated otherwise, the simulation parameters are summarized in~\cref{tab:parameters}. For the default channel setting, the settling time and equilibrium gain are $t^{\mathrm{eq}}=0.75$~s and $h^{\mathrm{eq}}=0.67$, respectively. 

\begin{table}[ht]
    \centering
    \caption{Simulation Parameters.}
    \label{tab:parameters}
    \begin{tabular}{@{}lccc@{}}
        \toprule
            \textbf{Parameter} & \textbf{Symbol} & \textbf{Value}    & \textbf{Unit} \\ 
            \midrule
            Diffusion coefficient of cDNAs  & $D$               & \num{150}    & \si{\micro\metre\squared\per\second} \\
            Number of released cDNAs        & $Q$               & \num{1000}       & -- \\
            Effective volumetric probe concentration         & $c_\mathrm{p}$      & \num{1}       & \si{\micro M} \\
            Probe surface density & $\sigma_{\mathrm p}$ & \num{5e-11} & \si{\mole\per\metre\squared} \\
            Number of free states                     & $N_\mathrm{f}$               & \num{100}        & -- \\
            Binding rate constant                & $k_{\mathrm{on}}$     & \num{6e8}        & \si{ M^{-1}.s^{-1}} \\
            Unbinding rate constant            & $k_{\mathrm{off}}$    & \num{3}          & \si{s^{-1}} \\
            Spatial step size                    & $\Delta x$        & \num{50}       & \si{\nano\meter} \\
            Time step size                       & $\Delta t$        & \num{8.25}    & \si{\micro s} \\
        \bottomrule
    \end{tabular}
\end{table}

\subsection{BER Versus Release Count $Q$}
\label{subsec:ber_vs_q}

 \begin{figure*}[t]
    \centering
    \subfigure[\label{fig:ber_vs_q_weak}Weak ISI, $T_\mathrm{b}=0.3$ s]{%
        \centering
        \includegraphics[width=0.45\linewidth]{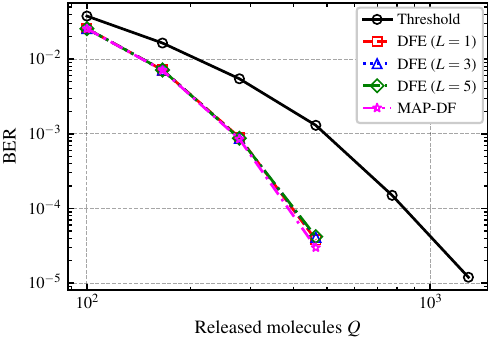}}
    \hfill
    \subfigure[\label{fig:ber_vs_q_strong}Strong ISI, $T_\mathrm{b}=0.1$ s]{%
        \centering
        \includegraphics[width=0.45\linewidth]{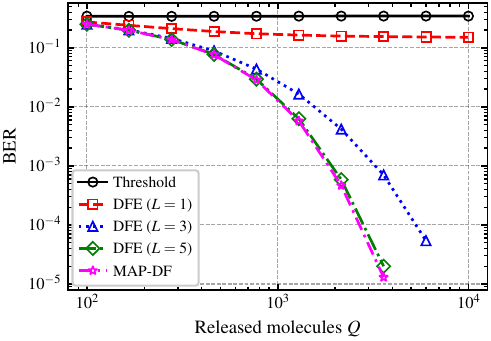}}
        \caption{BER versus $Q$ for the differential-threshold detector, finite-memory DFE with $L\in\{1,3,5\}$ feedback taps, and the MAP-DF benchmark in (a) weak-ISI ($T_\mathrm{b}=0.3$~s) and (b) strong-ISI ($T_\mathrm{b}=0.1$~s) settings.}
        \label{fig:ber_q}
\end{figure*}

The release count $Q$ controls the strength of the useful observation relative to the counting noise and therefore plays a role analogous to a signal-to-noise ratio (SNR) parameter in conventional MC detection problems~\cite{zheng2020noise,zheng2025anis}. Although the counting noise considered here is temporally correlated, increasing $Q$ still improves the BER performance in general by enhancing the desired differential signal relative to both residual ISI and noise.

\Cref{fig:ber_q} shows the BER as a function of $Q$ in two representative regimes. In~\cref{fig:ber_vs_q_weak}, the symbol interval is $T_\mathrm{b}=0.3$~s, which corresponds to a weak-memory regime. In this case, the differential-threshold detector already achieves good performance, and the DFE curves with $L=1,3,5$ are nearly indistinguishable. The MAP-DF curve also almost overlaps with the DFE curves. This behavior is consistent with~\cref{fig:cir_taps_tb_teq}: for large $T_\mathrm{b}$, the post-cursor taps $\{\Delta h_\ell\}_{\ell\ge1}$ are already small, so the dominant post-cursor ISI is captured by only the first few taps. Therefore, extending the feedback-memory length or using the likelihood-based MAP-DF rule provides only marginal additional gain.

In contrast, \cref{fig:ber_vs_q_strong} corresponds to a strong-memory regime with $T_\mathrm{b}=0.1$~s. Here, the differential-threshold detector and the DFE with $L=1$ remain interference-limited, especially at large $Q$. This is because the post-cursor ISI is not dominated by only the first post-cursor differential CIR tap; instead, several post-cursor taps remain significant. Consequently, canceling only the first tap is insufficient, whereas the DFEs with $L=3$ and $L=5$ achieve clear BER improvements. This behavior is consistent with~\cref{fig:delta_h-c}, where the post-cursor differential CIR tail extends over the first few symbol lags.

The comparison with MAP-DF further clarifies the performance-complexity tradeoff. MAP-DF achieves the best or nearly best BER performance because it exploits the conditional mean and variance of the differential observation. However, its improvement over DFE with $L=5$ is small, and the two curves almost overlap over most of the tested $Q$ range. This indicates that a finite-memory DFE with a properly selected feedback-memory length $L$ can capture most of the useful ISI-mitigation gain, while requiring a much simpler decision rule than MAP-DF. Therefore, the DFE provides an effective low-complexity receiver for the considered long-memory DNA-based MC channel.

Overall, the benefit of increasing the feedback-memory length $L$ depends on both the effective channel memory and the release count. When the channel memory is weak, a small $L$ already captures the dominant post-cursor ISI. When the memory is strong, a larger $L$ becomes beneficial because more post-cursor terms must be canceled. In addition, the gain from feedback also depends on decision reliability: for small $Q$, the symbol-rate observation is more noise-limited, so error propagation reduces the benefit of longer feedback. As $Q$ increases, previous decisions become more reliable, which improves ISI cancellation and makes the gain from larger $L$ more visible.

\subsection{BER Versus Symbol Interval $T_\mathrm{b}$}
\label{subsec:ber_vs_tb}

\begin{figure*}[t]
    \centering
    \subfigure[\label{fig:ber_vs_tb:a}$Q=1000$]{%
        \centering
        \includegraphics[width=0.45\linewidth]{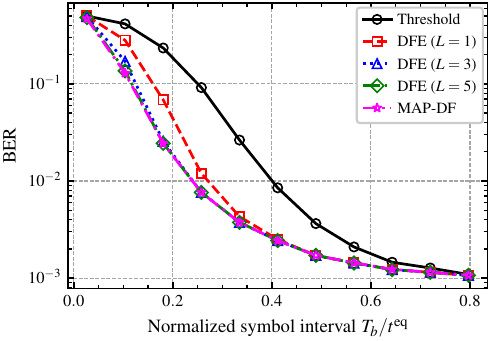}}
    \hfill
    \subfigure[\label{fig:ber_vs_tb:b}$Q=5000$]{%
        \centering
        \includegraphics[width=0.45\linewidth]{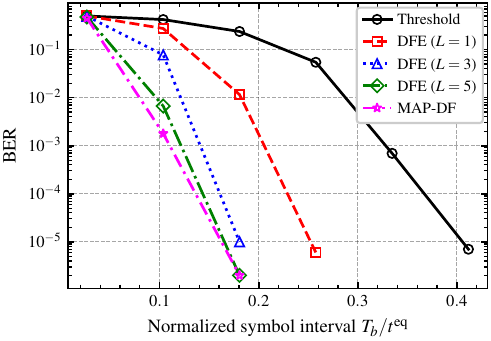}}
    \caption{BER versus the normalized symbol interval $T_\mathrm{b}/t^{\mathrm{eq}}$ for the differential-threshold detector, DFE with $L\in\{1,3,5\}$, and the MAP-DF benchmark.}
    \label{fig:ber_tb}
\end{figure*}

For a fixed channel setting, the severity of the differential-domain ISI is mainly governed by the normalized symbol interval $T_{\mathrm b}/t^{\mathrm{eq}}$; see~\cref{rem:design_Tb_teq}. The raw symbol-spaced CIR taps $h_\ell$ approach the nonzero equilibrium gain $h^{\mathrm{eq}}$ rather than decaying to zero. Therefore, for the differential receivers considered here, the interference contributed by previous symbols is governed by the post-cursor differential CIR taps $\{\Delta h_\ell\}_{\ell\ge 1}$; see~\cref{eq:delta_h_def}. Since $h_\ell$ approaches $h^{\mathrm{eq}}$ over the time scale characterized by $t^{\mathrm{eq}}$, increasing $T_{\mathrm b}$ allows more channel relaxation between consecutive symbol boundaries and causes the differential CIR taps to become negligible over fewer symbol lags. When $T_{\mathrm b}/t^{\mathrm{eq}}\ll 1$, several post-cursor differential CIR taps remain significant and the ISI is strong. In contrast, when $T_{\mathrm b}/t^{\mathrm{eq}}\gg 1$, these taps become small and the post-cursor ISI becomes weak.

\Cref{fig:ber_vs_tb:a} shows the BER for $Q=1000$, which corresponds to a relatively noise-limited operating point. In the small-$T_\mathrm{b}/t^{\mathrm{eq}}$ regime, the DFE consistently outperforms the differential-threshold detector. In this regime, the effective channel memory is long relative to the symbol interval, so previously released molecules continue to contribute over multiple symbol times and the post-cursor terms in~\cref{eq:delta_z_signal_isi} remain non-negligible. By explicitly estimating and canceling these terms through~\cref{eq:isi_hat_delta}, the DFE achieves a noticeable BER gain. The MAP-DF benchmark almost overlaps with $\mathrm{DFE}(L=3)$ and $\mathrm{DFE}(L=5)$ over most of the tested range. This indicates that, at this lower $Q$, the performance is still partly limited by counting noise and feedback decision errors; hence, the more statistics-aware MAP-DF rule provides only marginal additional gain over a finite-memory DFE with a few taps. As $T_\mathrm{b}/t^{\mathrm{eq}}$ increases, the post-cursor taps decay more rapidly, the residual ISI weakens, and the BER curves gradually converge. Hence, in the weak-memory regime, the simple differential-threshold detector becomes sufficient.
At this lower $Q$, the performance gap between $\mathrm{DFE}(L=3)$, $\mathrm{DFE}(L=5)$, and MAP-DF remains relatively small. This is mainly due to error propagation in the feedback loop: when the differential observation is noisy, the ISI estimate in~\cref{eq:isi_hat_delta} is formed from imperfect previous decisions, which limits the gain of additional feedback taps or likelihood-based refinement.

\Cref{fig:ber_vs_tb:b} shows the corresponding result for $Q=5000$, where the receiver operates in a less noise-limited regime. In this case, the BER is more strongly influenced by residual ISI than by counting noise. As expected, the BER of all detectors decreases as $T_\mathrm{b}/t^{\mathrm{eq}}$ increases, consistent with the shortening of the differential CIR taps in~\cref{fig:cir_taps_tb_teq}. Moreover, the DFE consistently outperforms the differential-threshold detector, and increasing $L$ improves performance, although with diminishing returns. This again reflects the fact that the dominant post-cursor ISI is carried by only the first few differential taps. The MAP-DF benchmark provides the best or nearly best performance, especially when $T_\mathrm{b}/t^{\mathrm{eq}}$ is small and the residual ISI is strong. However, the improvement over $\mathrm{DFE}(L=5)$ is limited, and the two curves become nearly indistinguishable as $T_\mathrm{b}/t^{\mathrm{eq}}$ increases. When $T_\mathrm{b}/t^{\mathrm{eq}}$ is sufficiently large, for example beyond approximately $0.2$ in this setting, no bit errors are observed for $\mathrm{DFE}(L=3)$, $\mathrm{DFE}(L=5)$, and MAP-DF within the simulated bit sequence, so the corresponding BER points are not visible on the logarithmic scale.

Taken together, the results in \cref{fig:ber_tb} show that the normalized symbol interval $T_\mathrm{b}/t^{\mathrm{eq}}$ provides a meaningful regime indicator for detector selection. When $T_\mathrm{b}/t^{\mathrm{eq}}$ is large, the channel is effectively weak-memory, and the differential-threshold detector provides a favorable low-complexity solution. When $T_\mathrm{b}/t^{\mathrm{eq}}$ is small, the effective channel memory becomes significant and finite-memory decision feedback is beneficial, especially when $Q$ is sufficiently large to keep error propagation under control. The comparison with MAP-DF further shows that a finite-memory DFE, particularly with $L=3$ or $L=5$, can approach the performance of a stronger statistics-aware benchmark while maintaining a simpler decision rule.

\subsection{Sensitivity to Physical Parameters}
\label{subsec:param_sensitivity}

\begin{figure*}[t]
    \includegraphics[width=0.9\linewidth]{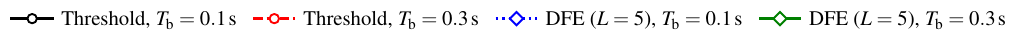}
    \centering
    \subfigure[\label{fig:param_sensitivity:a}Normalized diffusion coefficient]{%
        \centering
        \includegraphics[width=0.32\linewidth]{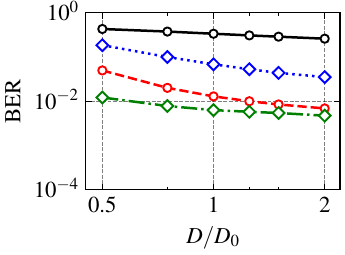}}
    \hfill
    \subfigure[\label{fig:param_sensitivity:b}Normalized association rate]{%
        \centering
        \includegraphics[width=0.32\linewidth]{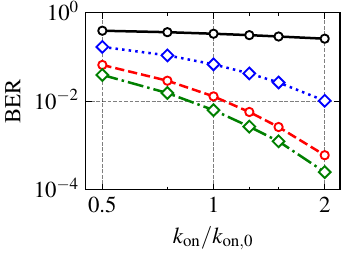}}
    \hfill
    \subfigure[\label{fig:param_sensitivity:c}Normalized dissociation rate]{%
        \centering
        \includegraphics[width=0.32\linewidth]{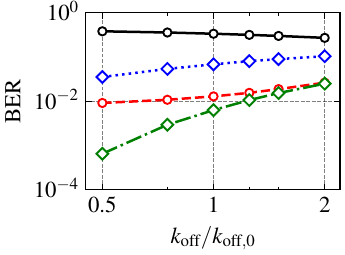}}
    \caption{BER sensitivity to variations in key physical parameters for $Q=1000$: (a) normalized diffusion coefficient $D/D_0$, (b) normalized association rate $k_{\mathrm{on}}/k_{\mathrm{on},0}$, and (c) normalized dissociation rate $k_{\mathrm{off}}/k_{\mathrm{off},0}$. Here, $D_0$, $k_{\mathrm{on},0}$, and $k_{\mathrm{off},0}$ denote the nominal parameter values specified in~\cref{tab:parameters}. Each parameter is varied individually while the remaining parameters are fixed at their nominal values.}
    \label{fig:param_sensitivity}
\end{figure*}

To further evaluate the robustness of the considered low-complexity receivers beyond the transmission parameters $Q$ and $T_{\mathrm b}$, we investigate their BER sensitivity to three key physical parameters: the diffusion coefficient $D$, association rate $k_{\mathrm{on}}$, and dissociation rate $k_{\mathrm{off}}$. For clarity, we consider the differential-threshold detector and DFE with $L=5$ as representative low-complexity receivers. The former provides a memoryless baseline without explicit ISI cancellation, whereas the latter captures most of the performance gain of finite-memory feedback observed in the preceding results while remaining substantially simpler than MAP-DF. Each physical parameter is varied individually from $0.5$ to $2$ times its nominal value while the remaining parameters are fixed. Here, $D_0$, $k_{\mathrm{on},0}$, and $k_{\mathrm{off},0}$ denote the nominal values specified in~\cref{tab:parameters}. The results for $Q=1000$ and $T_{\mathrm b}\in\{0.1,0.3\}$~s are shown in~\cref{fig:param_sensitivity}.

As shown in~\cref{fig:param_sensitivity:a}, the BER generally decreases as $D$ increases. This behavior is consistent with~\cref{eq:pd}, where $p_{\mathrm{diff}}=D\Delta t/(\Delta x)^2$, so a larger $D$ increases the per-step diffusion probability and accelerates molecular transport toward the RX. Since $D$ does not appear explicitly in the equilibrium gain in~\cref{eq:heq_microarray_closedform}, this improvement is mainly associated with faster transient transport rather than a change in the equilibrium bound fraction.

A stronger dependence is observed for $k_{\mathrm{on}}$ in~\cref{fig:param_sensitivity:b}. According to~\cref{eq:p_bind_general}, increasing $k_{\mathrm{on}}$ increases the binding probability $p_{\mathrm{bind}}$. Moreover, for fixed $k_{\mathrm{off}}$, the dissociation constant $K_{\mathrm D}=k_{\mathrm{off}}/k_{\mathrm{on}}$ decreases, and~\cref{eq:heq_microarray_closedform} therefore predicts a larger equilibrium gain $h^{\mathrm{eq}}$. The resulting stronger bound-state response improves the separability of the differential observations. Over the investigated range, this enhancement dominates the accompanying change in residual ISI, leading to a consistent BER reduction for both the differential-threshold detector and DFE as $k_{\mathrm{on}}$ increases.

The dependence on $k_{\mathrm{off}}$ in~\cref{fig:param_sensitivity:c} reflects a tradeoff between channel memory and bound-state response strength. From~\cref{eq:p_unbind_general}, increasing $k_{\mathrm{off}}$ increases the unbinding probability and accelerates relaxation of the bound population, thereby reducing the persistence of previous-symbol contributions. At the same time, for fixed $k_{\mathrm{on}}$, increasing $k_{\mathrm{off}}$ increases $K_{\mathrm D}$ and reduces $h^{\mathrm{eq}}$ according to~\cref{eq:heq_microarray_closedform}, corresponding to weaker molecular retention at the RX. For the differential-threshold detector at the shorter $T_{\mathrm b}=0.1$~s, where residual ISI is relatively strong, the benefit of the reduced channel memory outweighs the loss in bound-state response strength over the investigated range, leading to a gradual BER reduction as $k_{\mathrm{off}}$ increases. At $T_{\mathrm b}=0.3$~s, the residual ISI is already weaker, so the benefit of further reducing the channel memory becomes smaller and the reduced response strength becomes dominant, resulting in an increasing BER. For DFE, the dominant post-cursor ISI is already partially canceled through~\cref{eq:isi_hat_delta}; therefore, it obtains less additional benefit from the reduced channel memory and remains more strongly affected by the weaker bound-state response, leading to an overall BER increase with $k_{\mathrm{off}}$.

Overall, within the investigated ranges $0.5D_0\leq D\leq2D_0$, $0.5k_{\mathrm{on},0}\leq k_{\mathrm{on}}\leq2k_{\mathrm{on},0}$, and $0.5k_{\mathrm{off},0}\leq k_{\mathrm{off}}\leq2k_{\mathrm{off},0}$, the DFE consistently outperforms the differential-threshold detector for the considered operating points.

\section{Conclusion}
\label{sec:conclusion}
\subsection{Summary}
This paper developed a Markov state-space framework for DNA-based MC with microarray-style reception under reversible hybridization. Starting from a voxelized reaction-diffusion model, we constructed a block-structured transition matrix that captures molecular transport and binding/unbinding dynamics in a communication-oriented form. For the considered 1D microarray specialization, the proposed framework yields the CIR, the equilibrium gain, and a settling-time-based characterization of the effective channel memory. Building on the resulting symbol-rate observation model for OOK, we further derived a grouped-binomial counting model and a closed-form expression for the covariance of the counting noise, thereby explicitly characterizing its colored nature. Based on these channel and noise statistics, we developed two low-complexity receivers, namely a differential-threshold detector and a finite-memory DFE. Numerical results showed that the relative performance of the proposed receivers depends strongly on the channel-memory regime, and that the normalized symbol interval $T_\mathrm{b}/t^{\mathrm{eq}}$ provides a meaningful indicator for detector selection. Overall, the proposed framework provides a principled bridge from physical channel dynamics to communication-level characterization and practical RX design for DNA microarray channels, and offers a useful analytical basis for DNA-enabled sensing and communication in emerging IoBNT systems. 

\subsection{Future Research Directions}
Future work can proceed along several directions that are closely related to emerging IoT and IoBNT systems. The derived covariance structure can be exploited in more advanced receiver designs, including covariance-aware detectors, finite-memory sequence detectors, and adaptive low-complexity receivers for different channel-memory regimes, which is particularly relevant for resource-constrained bio-nano devices. The proposed framework can also be extended beyond the present 1D specialization to higher-dimensional advection-diffusion settings and more general boundary conditions, enabling more realistic modeling of microfluidic platforms, lab-on-chip systems, and in-body molecular signaling environments. In addition, the transition-matrix formulation can be generalized to multi-observable and multi-terminal architectures, thereby supporting the analysis of MIMO-like, multi-user, and networked molecular communication scenarios. A further promising direction is to incorporate joint sensing and communication functionality, since microarray-style receivers naturally serve as both molecular sensing interfaces and communication endpoints. This may open the way to system-level studies on receiver placement, symbol interval design, release count allocation, and robustness under biochemical uncertainty, and ultimately support the development of networked bio-nano sensing and communication platforms for future IoT and IoBNT applications.

\appendices

\section{Proof of Proposition~\ref{prop:mixing_rate}}
\label{app:proof_prop_mixing}

We prove the claim for the continuous-time surrogate $\mathbf A=(\mathbf P-\mathbf I)/\Delta t$
in~\cref{eq:A_def}. Assume for simplicity that $\mathbf P$ is diagonalizable.\footnote{If $\mathbf P$ is not diagonalizable, the same conclusion follows from its Jordan form; the additional polynomial factors do not change the exponential rate.}
Thus,
\begin{equation}
  \mathbf P=\mathbf V\boldsymbol\Lambda\mathbf V^{-1},
  \qquad
  \boldsymbol\Lambda=\mathrm{diag}(\lambda_0,\lambda_1,\ldots,\lambda_{N-1}),
  \label{eq:app_P_diag}
\end{equation}
where $\lambda_0=1$ and $|\lambda_i|<1$ for $i\ge 1$. Let
\[
|\lambda_1|=\max_{i\ge 1}|\lambda_i|
\]
denote the SLEM of $\mathbf P$. Using~\cref{eq:A_def}, we obtain
\begin{equation}
  \mathbf A
  =\frac{\mathbf P-\mathbf I}{\Delta t}
  =\mathbf V\left(\frac{\boldsymbol\Lambda-\mathbf I}{\Delta t}\right)\mathbf V^{-1}.
  \label{eq:app_A_diag}
\end{equation}
Hence,
\begin{equation}
\exp(\mathbf A t)
=
\mathbf V\,
\mathrm{diag}\!\left(
1,\,
e^\frac{(\lambda_1-1)t}{\Delta t},\ldots,
e^\frac{(\lambda_{N-1}-1)t}{\Delta t}
\right)\mathbf V^{-1}.
\label{eq:app_expAt}
\end{equation}
Since $|\lambda_i|<1$ for $i\ge 1$, we have $\Re(\lambda_i)-1<0$, and therefore
\[
e^{(\lambda_i-1)t/\Delta t}\to 0,\qquad t\to\infty .
\]
Thus, the equilibrium component is
\begin{equation}
\bm x^{\mathrm{eq}}
=
\mathbf V\,\mathrm{diag}(1,0,\ldots,0)\mathbf V^{-1}\bm x(0),
\label{eq:app_xeq}
\end{equation}
and
\begin{equation}
\bm x(t)-\bm x^{\mathrm{eq}}
=
\mathbf V\,\mathrm{diag}\!\left(
0,\,
e^\frac{(\lambda_1-1)t}{\Delta t},\ldots,
e^\frac{(\lambda_{N-1}-1)t}{\Delta t}
\right)\mathbf V^{-1}\bm x(0).
\label{eq:app_error_vec}
\end{equation}
Taking the $\infty$-norm and using submultiplicativity yields
\begin{equation}
\begin{aligned}
&\|\bm x(t)-\bm x^{\mathrm{eq}}\|_\infty
\le
\|\mathbf V\|_\infty\\
&\times
\left\|
\mathrm{diag}\!\left(
0,\,
e^{\frac{(\lambda_1-1)t}{\Delta t}},\ldots,
e^{\frac{(\lambda_{N-1}-1)t}{\Delta t}}
\right)
\right\|_\infty 
\|\mathbf V^{-1}\|_\infty\,\|\bm x(0)\|_\infty \\
&=
\kappa_\infty(\mathbf V)\,
\max_{i\ge 1}\left|e^{\frac{(\lambda_i-1)t}{\Delta t}}\right|
\,\|\bm x(0)\|_\infty .
\end{aligned}
\label{eq:app_norm_bound}
\end{equation}
where $\kappa_\infty(\mathbf V)\triangleq \|\mathbf V\|_\infty\|\mathbf V^{-1}\|_\infty$. Moreover, for each $i\ge 1$,
\begin{equation}
\left|e^\frac{(\lambda_i-1)t}{\Delta t}\right|
=
e^\frac{(\Re\{\lambda_i\}-1)t}{\Delta t}
\le
e^\frac{(|\lambda_i|-1)t}{\Delta t}
\le
e^\frac{-(1-|\lambda_1|)t}{\Delta t},
\label{eq:app_slem_step}
\end{equation}
where we used $\Re\{\lambda_i\}\le |\lambda_i|$ and $|\lambda_i|\le |\lambda_1|$. Substituting~\cref{eq:app_slem_step} into~\cref{eq:app_norm_bound} gives
\begin{equation}
\|\bm x(t)-\bm x^{\mathrm{eq}}\|_\infty
\le
C\,\exp\!\left(-\frac{(1-|\lambda_1|)t}{\Delta t}\right),
\label{eq:app_bound_pre_tau}
\end{equation}
where
\[
C\triangleq \kappa_\infty(\mathbf V)\|\bm x(0)\|_\infty .
\]
Invoking~\cref{eq:tau_def} in~\cref{eq:app_bound_pre_tau} yields
\begin{equation}
\|\bm x(t)-\bm x^{\mathrm{eq}}\|_\infty
\le
C\,\exp\!\left(-\frac{t}{\tau}\right),
\end{equation}
which proves~\cref{prop:mixing_rate}.

\begin{figure*}[t]
    \centering
    \subfigure[\label{fig:gaussian_validation:a}$Q=50$]{%
        \centering
        \includegraphics[width=0.45\linewidth]{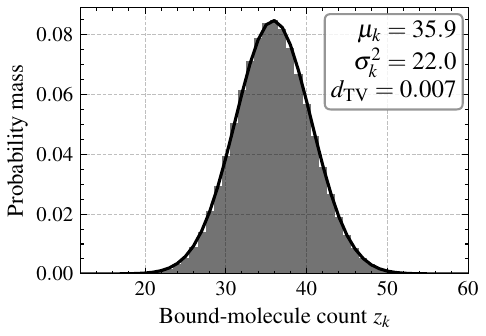}}
    \hfill
    \subfigure[\label{fig:gaussian_validation:b}$Q=5000$]{%
        \centering
        \includegraphics[width=0.45\linewidth]{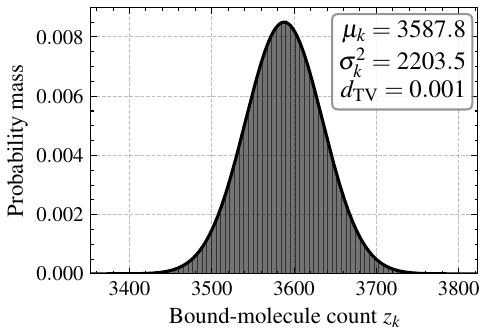}}
    \caption{Comparison between the grouped-binomial model in~\cref{eq:binom_batch_dt,eq:z_superposition_dt} and the corresponding Gaussian approximation under low- and high-count conditions. The bars show empirical probability masses obtained from $5\times10^6$ Monte Carlo realizations, whereas the solid curves show the continuity-corrected Gaussian probabilities in~\cref{eq:app_gaussian_mass}. The conditional mean, conditional variance, and total variation distance are reported in each panel.}
    \label{fig:gaussian_validation}
   \vspace{-0.3cm}
\end{figure*}

\section{Proof of Proposition~\ref{prop:cov_microarray}}
\label{app:proof_cov_microarray}

This appendix proves~\cref{prop:cov_microarray}. Throughout the proof, the transmit sequence $\{a_k\}$ is treated as fixed, so that $\bar z_k$ is deterministic and
\[
\Cov(w_k,w_{k-\ell})=\Cov(z_k,z_{k-\ell}).
\]

\subsection{Batch-wise decomposition}
Sampling~\cref{eq:z_superposition_dt} at symbol boundaries $n=kN_s$ yields
\begin{equation}
    z_k
    = \sum_{r=0}^{k} z^{(r)}[kN_s]
    \triangleq \sum_{r=0}^{k} z_k^{(r)},
    \label{eq:app_zk_batch_decomp}
\end{equation}
where $z_k^{(r)}$ denotes the contribution of the $r$-th release batch to the symbol-time observation $z_k$.

For the $r$-th batch, the batch size is $Qa_r$. For each molecule $m\in\{1,\ldots,Qa_r\}$, let
\begin{equation}
    R_{r,m}[n]\in\mathcal S,\qquad n\ge rN_s,
    \label{eq:app_Rrm_state}
\end{equation}
denote its Markov state, with initialization
\begin{equation}
    R_{r,m}[rN_s]=s_{i^{\mathrm{TX}}}.
    \label{eq:app_Rrm_init}
\end{equation}
Define the bound-state indicator at symbol time $k$ as
\begin{equation}
    I_{r,m}[k]\triangleq \mathbbm{1}\{R_{r,m}[kN_s]=s_N\}.
    \label{eq:app_indicator_def}
\end{equation}
Then
\begin{equation}
    z_k^{(r)}=\sum_{m=1}^{Qa_r} I_{r,m}[k].
    \label{eq:app_zkr_sum}
\end{equation}

\subsection{Covariance derivation}
For $k\ge \ell\ge 1$, we have
\begin{equation}
\begin{aligned}
    \Cov(z_k,z_{k-\ell})
    &=
    \Cov\!\left(\sum_{r=0}^{k} z_k^{(r)},\,\sum_{r'=0}^{k-\ell} z_{k-\ell}^{(r')}\right) \\
    &=
    \sum_{r=0}^{k}\sum_{r'=0}^{k-\ell}
    \Cov\!\left(z_k^{(r)},z_{k-\ell}^{(r')}\right).
\end{aligned}
\label{eq:app_cov_expand}
\end{equation}
Under the non-saturating probe assumption, different batches correspond to disjoint molecule sets and are independent. Hence,
\[
\Cov\!\left(z_k^{(r)},z_{k-\ell}^{(r')}\right)=0,\qquad r\neq r'.
\]
Moreover, $z_{k-\ell}^{(r)}=0$ for $r>k-\ell$. Therefore,
\begin{equation}
    \Cov(z_k,z_{k-\ell})
    =\sum_{r=0}^{k-\ell}\Cov\!\left(z_k^{(r)},z_{k-\ell}^{(r)}\right).
    \label{eq:app_cov_batch_sum}
\end{equation}

Using~\cref{eq:app_zkr_sum}, we obtain
\begin{equation}
\begin{aligned}
\Cov\!\left(z_k^{(r)},z_{k-\ell}^{(r)}\right)
&=
\Cov\!\left(\sum_{m=1}^{Qa_r} I_{r,m}[k],\,\sum_{m'=1}^{Qa_r} I_{r,m'}[k-\ell]\right) \\
&=
\sum_{m=1}^{Qa_r}\sum_{m'=1}^{Qa_r}
\Cov\!\left(I_{r,m}[k],I_{r,m'}[k-\ell]\right).
\end{aligned}
\label{eq:app_cov_within_batch_expand}
\end{equation}
Since molecules within the same batch are independent and identically distributed, the cross-molecule terms vanish for $m\neq m'$, and thus
\begin{equation}
\Cov\!\left(z_k^{(r)},z_{k-\ell}^{(r)}\right)
=
Qa_r\,\Cov\!\left(I_{r,1}[k],I_{r,1}[k-\ell]\right).
\label{eq:app_cov_within_batch}
\end{equation}

Because the $r$-th batch is released at symbol time $r$, the marginal bound probabilities depend only on the elapsed symbol intervals:
\begin{equation}
    \E\{I_{r,1}[k]\}=h_{k-r},
    \qquad
    \E\{I_{r,1}[k-\ell]\}=h_{k-\ell-r}.
    \label{eq:app_E_indicator_k}
\end{equation}
Moreover, for $r\le k-\ell$, the event $I_{r,1}[k-\ell]=1$ is equivalent to
\(
R_{r,1}[(k-\ell)N_s]=s_N
\).
Hence, by the time-homogeneous Markov property,
\begin{equation}
\begin{aligned}
&\Pr\!\left(I_{r,1}[k]=1\mid I_{r,1}[k-\ell]=1\right)\\
&=
\Pr\!\left(R_{r,1}[kN_s]=s_N \mid R_{r,1}[(k-\ell)N_s]=s_N\right) \\
&= (\mathbf P^{\ell N_s})_{N,N}
= \pi_N^{(\ell)}.
\end{aligned}
\label{eq:app_cond_prob_return}
\end{equation}
Therefore,
\begin{equation}
\Pr\!\left(I_{r,1}[k]=1,\,I_{r,1}[k-\ell]=1\right)
=
h_{k-\ell-r}\,\pi_N^{(\ell)}.
\label{eq:app_joint_prob}
\end{equation}
Since the indicators are binary,
\begin{equation}
\begin{aligned}
&\Cov\!\left(I_{r,1}[k],I_{r,1}[k-\ell]\right)\\
&=
\Pr\!\left(I_{r,1}[k]=1,\,I_{r,1}[k-\ell]=1\right)
-\E\{I_{r,1}[k]\}\E\{I_{r,1}[k-\ell]\} \\
&=
h_{k-\ell-r}\bigl(\pi_N^{(\ell)}-h_{k-r}\bigr).
\end{aligned}
\label{eq:app_cov_indicator}
\end{equation}

Substituting~\cref{eq:app_cov_indicator} into~\cref{eq:app_cov_within_batch} and then into~\cref{eq:app_cov_batch_sum} gives
\begin{equation}
    \Cov(z_k,z_{k-\ell})
    =
    Q\sum_{r=0}^{k-\ell} a_r\,h_{k-\ell-r}\bigl(\pi_N^{(\ell)}-h_{k-r}\bigr).
    \label{eq:app_cov_final_z}
\end{equation}
Since $\Cov(w_k,w_{k-\ell})=\Cov(z_k,z_{k-\ell})$, this proves~\cref{prop:cov_microarray}.

\section{Proof of~\cref{cor:lag_decorrelation}}
\label{app:proof_cor_lag_decorrelation_tau}

Under the continuous-time surrogate, define
\begin{equation}
\bm x^{(q)}(t)\triangleq \exp(\mathbf A t)\bm e_q .
\end{equation}
By~\cref{prop:mixing_rate}, there exists a uniform constant $C_\star>0$ such that
\begin{equation}
\|\bm x^{(q)}(t)-\bm x^{\mathrm{eq}}\|_\infty
\le C_\star e^{-t/\tau},
\qquad \forall q,\ t\ge 0.
\end{equation}
Taking the $N$-th component yields
\begin{equation}
\bigl|\bm e_N^\top \exp(\mathbf A t)\bm e_q - x_N^{\mathrm{eq}}\bigr|
\le C_\star e^{-t/\tau}.
\end{equation}
Using
\begin{equation}
\pi_N^{(\ell)}=\bm e_N^\top \exp(\mathbf A \ell T_\mathrm{b})\bm e_N,
\qquad
h_i=\bm e_N^\top \exp(\mathbf A i T_\mathrm{b})\bm e_1,
\end{equation}
we obtain, for all $i\ge \ell$,
\begin{equation}
|\pi_N^{(\ell)}-x_N^{\mathrm{eq}}|
\le C_\star e^{-\ell T_\mathrm{b}/\tau},
\,
|h_i-x_N^{\mathrm{eq}}|
\le C_\star e^{-iT_\mathrm{b}/\tau}
\le C_\star e^{-\ell T_\mathrm{b}/\tau}.
\end{equation}
Hence,
\begin{equation}
|\pi_N^{(\ell)}-h_i|
\le
|\pi_N^{(\ell)}-x_N^{\mathrm{eq}}|
+
|h_i-x_N^{\mathrm{eq}}|
\le
2C_\star e^{-\ell T_\mathrm{b}/\tau}.
\end{equation}

Substituting this bound into~\cref{prop:cov_microarray} and using
\begin{equation}
0\le h_{k-\ell-r}\le 1,
\qquad
a_r\in\{0,1\},
\end{equation}
we obtain
\begin{equation}
|\Cov(w_k,w_{k-\ell})|
\le
2Q C_\star (k-\ell+1)e^{-\ell T_\mathrm{b}/\tau}.
\end{equation}
Therefore,~\cref{eq:cov_decay_bound_tau} holds with $C=2Q C_\star$.

\section{Validation of the Gaussian approximation}
\label{app:gaussian_validation}

This appendix validates the Gaussian approximation adopted in~\cref{eq:gaussian_approximation} against the original grouped-binomial model defined in~\cref{eq:binom_batch_dt,eq:z_superposition_dt}.
For a fixed transmitted sequence $\boldsymbol{a}$, the exact probability mass function (PMF) of the grouped-binomial observation is defined as
\begin{equation}
p_k(m)
\triangleq
\Pr\{z_k=m\},
\label{eq:exact_pmf}
\end{equation}
where $z_k$ is the grouped-binomial observation defined in~\cref{eq:wk_def}, and $m$ is a nonnegative integer in the support of $z_k$. To compare the continuous Gaussian approximation with the discrete observation model, the Gaussian probability associated with each integer $m\in\mathbb{Z}$ is evaluated over a unit-width interval centered at $m$, yielding
\begin{equation}
\widetilde p_k(m)
\triangleq
\Phi\left(
\frac{m+\frac{1}{2}-\bar z_k}{\sigma_k}
\right)
-
\Phi\left(
\frac{m-\frac{1}{2}-\bar z_k}{\sigma_k}
\right),
\label{eq:app_gaussian_mass}
\end{equation}
where $\Phi(\cdot)$ is the standard Gaussian cumulative distribution function, and $\bar z_k$ and $\sigma_k^2$ are given in~\cref{eq:zk_bar_symbol_isi,eq:var_z_symbol}, respectively.

The approximation error is quantified by the total variation distance between the exact grouped-binomial PMF and the discretized Gaussian PMF~\cite{gibbs2002choosing},
\begin{equation}
d_{\mathrm{TV}}
\triangleq
\frac{1}{2}
\sum_{m\in\mathbb{Z}}
\left|
p_k(m)-\widetilde p_k(m)
\right|,
\label{eq:app_dtv}
\end{equation}
where $p_k(m)=0$ for $m<0$. Smaller values of $d_{\mathrm{TV}}$ indicate closer agreement between the two models. 

For the validation, the diffusion and reaction parameters are kept identical to those used in~\cref{tab:parameters}. We consider $T_\mathrm{b}=0.1~\mathrm{s}$, the transmitted sequence $\boldsymbol{a}=[1,1,1]$, and the observation at $k=2$, corresponding to the third symbol interval. The low- and high-count conditions correspond to $Q=50$ and $Q=5000$, respectively. For each condition, the empirical probability masses are estimated from $5\times10^6$ independent Monte Carlo realizations. As shown in~\cref{fig:gaussian_validation:a}, the Gaussian approximation captures the grouped-binomial distribution reasonably well even for the low-count case $Q=50$, yielding $d_{\mathrm{TV}}\approx 0.007$, despite a small mismatch caused by the discreteness and asymmetry of the molecule count. For the high-count case $Q=5000$ in~\cref{fig:gaussian_validation:b}, the two distributions nearly overlap, with $d_{\mathrm{TV}}\approx 0.001$.

A binomial count is generally well approximated by a Gaussian distribution when both its expected observed and unobserved molecule counts are sufficiently large. In the proposed grouped-binomial model, the additional aggregation of contributions from multiple release batches further promotes Gaussian-like behavior. In the proposed grouped-binomial model, the observation further aggregates the contributions of multiple release batches. Consequently, even the considered case with $Q=50$ provides a sufficiently accurate Gaussian approximation under the adopted operating parameters, supporting the use of~\cref{eq:gaussian_approximation} in place of the original grouped-binomial model.

\bibliographystyle{IEEEtran}
\bibliography{IEEEabrv,refs}

\clearpage
\onecolumn
\makeatletter
\let\section\@IEEEappendixsavesection
\makeatother
\setcounter{section}{0}
\renewcommand{\thesection}{S\arabic{section}}
\renewcommand{\thesectiondis}{\thesection.}
\renewcommand{\theHsection}{supp.\arabic{section}}
\renewcommand{\theHequation}{supp.\arabic{equation}}
\renewcommand{\theHfigure}{supp.\arabic{figure}}
\renewcommand{\theHtable}{supp.\arabic{table}}
\pdfbookmark[0]{Supplementary Material}{supplementary}

\linespread{1.3}\selectfont

\renewcommand{\theequation}{S\arabic{equation}}
\renewcommand{\thefigure}{S\arabic{figure}}
\renewcommand{\thetable}{S\arabic{table}}
\setcounter{equation}{0}
\setcounter{figure}{0}
\setcounter{table}{0}

\begin{center}
\begin{large}
\textbf{SUPPLEMENTARY MATERIAL}\\
\textbf{Markov State-Space Modeling and Channel Characterization for DNA-Based\\ Molecular Communication}\\
\textbf{Ruifeng~Zheng, Zhihan Xu, Veronika Volkova, Pengjie Zhou, Martín Schottlender,\\ Juan A. Cabrera, Frank H.\,P. Fitzek, and Pit Hofmann}
\end{large}
\end{center}

\section{Two-dimensional (2D) example}
\label{sec:supp_2d}
This section presents an illustrative two-dimensional (2D) construction demonstrating that the proposed Markov framework extends beyond the one-dimensional (1D) specialization considered in the main manuscript. A compact $5\times5$ free-voxel grid is used to explicitly describe the state-space construction, partial receiver coverage, transition matrix, and receiver observation. The example is intended to illustrate the multidimensional extensibility of the framework and to verify the resulting matrix-based predictions through stochastic simulation, rather than to reproduce a specific experimental geometry.

\subsection{Geometry and State Definition}

We consider a bounded square domain discretized into a $5\times5$ grid of free diffusion voxels. The free state associated with the voxel in row $i$ and column $j$ is denoted by $s_{i,j}$, where $i,j\in\{1,2,3,4,5\}$. The rows are indexed from top to bottom, and the columns are indexed from left to right. The corresponding 1D free-state index is
\begin{equation}
\nu(i,j)\triangleq5(i-1)+j.
\label{eq:supp_state_index}
\end{equation}

The transmitter is located at the central voxel $s_{\mathrm{TX}}=s_{3,3}$. The three central voxels along the bottom boundary,
\begin{equation}
\mathcal V_{\mathrm{RX}}=\left\{s_{5,2},s_{5,3},s_{5,4}\right\},
\label{eq:supp_rx_adjacent_set}
\end{equation}
are adjacent to the reactive receiver boundary. Each receiver-adjacent free voxel is connected to one local bound state, namely, $s_{5,2}\leftrightarrow b_1$, $s_{5,3}\leftrightarrow b_2$, and $s_{5,4}\leftrightarrow b_3$. Thus, only the three central segments of the bottom boundary are reactive, illustrating partial receiver coverage, while all remaining boundary segments are reflecting. The resulting model contains $N_{\mathrm f}=25$ free states and $N_{\mathrm b}=3$ bound states, for a total of $N_{\mathrm{2D}}=28$ Markov states.

\begin{figure}[t]
\centering
\includegraphics[width=0.45\linewidth]{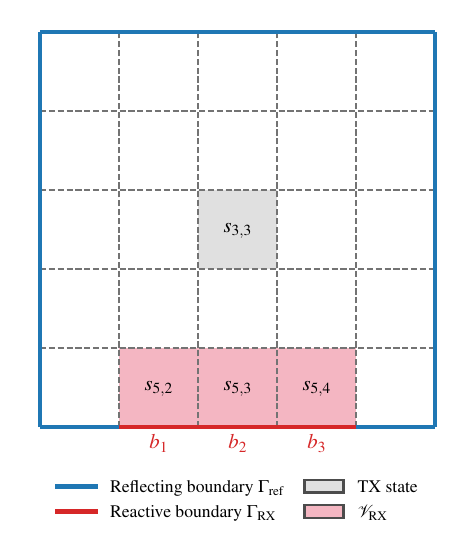}
\caption{Illustrative 2D geometry based on a $5\times5$ free-state grid. The transmitter is located at the central voxel $s_{3,3}$. The three bottom-center free voxels $s_{5,2}$, $s_{5,3}$, and $s_{5,4}$ are adjacent to the reactive receiver boundary and are connected to bound states $b_1$, $b_2$, and $b_3$, respectively. All remaining boundary segments are reflecting.}
\label{fig:supp_2d_geometry}
\end{figure}

\subsection{Two-Dimensional Transition Matrix}

The diffusion probability from a free voxel to each horizontally or vertically adjacent free voxel is
\begin{equation}
p_{\mathrm{diff}}\triangleq\frac{D\Delta t}{(\Delta x)^2}.
\label{eq:supp_pdiff}
\end{equation}
For each free state $s_{i,j}$, the set of valid neighboring free states is
\begin{equation}
\mathcal N_{i,j}\triangleq\left\{s_{m,n}:|m-i|+|n-j|=1,\;1\leq m,n\leq5\right\}.
\label{eq:supp_neighbor_set}
\end{equation}
The binding and unbinding probabilities are given by
\begin{align}
p_{\mathrm{bind}}&\triangleq1-\exp\left(-k_{\mathrm{on}}c_{\mathrm p}\Delta t\right)\approx k_{\mathrm{on}}c_{\mathrm p}\Delta t,
\label{eq:supp_pbind}\\
p_{\mathrm{unbind}}&\triangleq1-\exp\left(-k_{\mathrm{off}}\Delta t\right)\approx k_{\mathrm{off}}\Delta t.
\label{eq:supp_punbind}
\end{align}
Here, $c_{\mathrm p}$ retains the interpretation introduced in the main manuscript, namely, the effective volumetric concentration of available immobilized probes in a receiver-adjacent reactive voxel. The state vector is written as
\begin{equation}
\mathbf x_{\mathrm{2D}}[k]=
\begin{bmatrix}
\mathbf x_{\mathrm f}^{\mathsf T}[k] & \mathbf x_{\mathrm b}^{\mathsf T}[k]
\end{bmatrix}^{\mathsf T}
\in\mathbb R^{28},
\label{eq:supp_state_vector}
\end{equation}
where $\mathbf x_{\mathrm f}[k]\in\mathbb R^{25}$ and $\mathbf x_{\mathrm b}[k]\in\mathbb R^3$ describe the free-state and bound-state populations, respectively.

The corresponding transition matrix has the block structure
\begin{equation}
\mathbf P_{\mathrm{2D}}=
\begin{bmatrix}
\mathbf Q_{\mathrm{2D}} & \mathbf U\\
\mathbf B & \mathbf R
\end{bmatrix}
\in\mathbb R^{28\times28},
\label{eq:supp_block_matrix}
\end{equation}
where $\mathbf Q_{\mathrm{2D}}$ describes diffusion among the free states, $\mathbf B$ describes binding, $\mathbf U$ describes unbinding, and $\mathbf R$ describes bound-state retention.

For two neighboring free states, the nonzero off-diagonal entries of $\mathbf Q_{\mathrm{2D}}$ are
\begin{equation}
[\mathbf Q_{\mathrm{2D}}]_{\nu(m,n),\nu(i,j)}=p_{\mathrm{diff}},\qquad s_{m,n}\in\mathcal N_{i,j}.
\label{eq:supp_free_transition}
\end{equation}
The self-transition probability of a free state is
\begin{equation}
[\mathbf Q_{\mathrm{2D}}]_{\nu(i,j),\nu(i,j)}
=1-|\mathcal N_{i,j}|p_{\mathrm{diff}}
-\mathbbm{1}_{\{s_{i,j}\in\mathcal V_{\mathrm{RX}}\}}p_{\mathrm{bind}},
\label{eq:supp_self_transition}
\end{equation}
where $\mathbbm{1}_{\{\cdot\}}$ denotes the indicator function. At a reflecting boundary, the outward diffusion transitions that would leave the domain are incorporated into the self-transition probability.

Let $j_r=r+1$ for $r\in\{1,2,3\}$, such that bound state $b_r$ corresponds to the receiver-adjacent free state $s_{5,j_r}$. The nonzero entries of the binding and unbinding blocks are
\begin{align}
[\mathbf B]_{r,\nu(5,j_r)}&=p_{\mathrm{bind}},
\label{eq:supp_binding_block}\\
[\mathbf U]_{\nu(5,j_r),r}&=p_{\mathrm{unbind}}.
\label{eq:supp_unbinding_block}
\end{align}
The bound-state transition block is
\begin{equation}
\mathbf R=(1-p_{\mathrm{unbind}})\mathbf I_3.
\label{eq:supp_bound_block}
\end{equation}
All entries not specified above are zero.

The time step $\Delta t$ must be selected such that all one-step transition probabilities remain nonnegative. Since an interior free voxel has four free neighbors, whereas each receiver-adjacent free voxel has three free neighbors and one binding transition, a sufficient condition is
\begin{equation}
\max\left\{4p_{\mathrm{diff}},\,3p_{\mathrm{diff}}+p_{\mathrm{bind}}\right\}\leq1.
\label{eq:supp_validity}
\end{equation}
Under this condition, every column of $\mathbf P_{\mathrm{2D}}$ sums to one, and the resulting transition matrix is column-stochastic.

\subsection{Receiver Observation and Equilibrium}

The receiver observation is associated with the total bound-state population. For the single-molecule state-probability vector, the observation vector is defined as
\begin{equation}
\mathbf o_{\mathrm{2D}}=
\begin{bmatrix}
\mathbf 0_{25}^{\mathsf T} & 1 & 1 & 1
\end{bmatrix}^{\mathsf T}.
\label{eq:supp_observation_vector}
\end{equation}

For a molecule initially released from the central transmitter state, the probability that it occupies one of the three bound states at time step $k$ is
\begin{equation}
h_{\mathrm{2D}}[k]=\mathbf o_{\mathrm{2D}}^{\mathsf T}\mathbf P_{\mathrm{2D}}^k\mathbf e_{\nu(3,3)},
\label{eq:supp_cir}
\end{equation}
where $\mathbf e_{\nu(3,3)}$ is the standard basis vector associated with the transmitter state.

If $Q$ molecules are independently released from the transmitter state, the expected total number of bound molecules is
\begin{equation}
\overline N_{\mathrm b}[k]=Qh_{\mathrm{2D}}[k].
\label{eq:supp_bound_number}
\end{equation}

The three local receiver responses can also be evaluated separately. Let $\mathbf e_{b_r}\in\mathbb R^{28}$ denote the standard basis vector corresponding to bound state $b_r$. The local binding probability and the corresponding expected bound-molecule count are
\begin{equation}
h_{b_r}[k]=\mathbf e_{b_r}^{\mathsf T}\mathbf P_{\mathrm{2D}}^k\mathbf e_{\nu(3,3)},\qquad
\overline N_{b_r}[k]=Qh_{b_r}[k],\qquad r\in\{1,2,3\}.
\label{eq:supp_local_response}
\end{equation}

Because the geometry and transmitter placement are symmetric about the vertical centerline,
\begin{equation}
h_{b_1}[k]=h_{b_3}[k],
\label{eq:supp_symmetry}
\end{equation}
which provides a direct consistency check for the analytical and stochastic implementations.

Because the free-state random walk is connected and symmetric and all three reactive sites have identical binding and unbinding probabilities, all free states have the same equilibrium probability, while the three bound states also have equal equilibrium probabilities. The total equilibrium bound-state probability is therefore
\begin{equation}
h_{\mathrm{2D}}^{\mathrm{eq}}=
\frac{3p_{\mathrm{bind}}}{25p_{\mathrm{unbind}}+3p_{\mathrm{bind}}}.
\label{eq:supp_equilibrium}
\end{equation}
The corresponding expected equilibrium number of bound molecules is
\begin{equation}
\overline N_{\mathrm b}^{\mathrm{eq}}=Qh_{\mathrm{2D}}^{\mathrm{eq}}.
\label{eq:supp_equilibrium_number}
\end{equation}

\subsection{Numerical Verification}

\begin{table}[t]
\caption{Simulation parameters for the illustrative 2D example.}
\label{tab:supp_parameters}
\centering
\small
\renewcommand{\arraystretch}{1.15}
\begin{tabular}{c p{0.53\linewidth} c}
\toprule
Symbol & Description & Value\\
\midrule
$D$ & Diffusion coefficient & $150~\mu\mathrm{m}^2/\mathrm{s}$\\
$\Delta x$ & Voxel side length & $50~\mathrm{nm}$\\
$\Delta t$ & Time-step duration & $4~\mu\mathrm{s}$\\
$Q$ & Number of released cDNAs & $1000$\\
$c_{\mathrm p}$ & Effective volumetric probe concentration & $1~\mu\mathrm{M}$\\
$k_{\mathrm{on}}$ & Binding rate constant & $6\times10^8~\mathrm{M}^{-1}\mathrm{s}^{-1}$\\
$k_{\mathrm{off}}$ & Unbinding rate constant & $3~\mathrm{s}^{-1}$\\
$N_{\mathrm f}$ & Number of free states & $25$\\
$N_{\mathrm b}$ & Number of bound states & $3$\\
\bottomrule
\end{tabular}
\end{table}
The numerical parameters used for the illustrative 2D example are summarized in~\cref{tab:supp_parameters}.
These parameters yield $p_{\mathrm{diff}}=0.24$, $p_{\mathrm{bind}}\approx2.397\times10^{-3}$, and $p_{\mathrm{unbind}}\approx1.200\times10^{-5}$. Therefore, $\max\{4p_{\mathrm{diff}},3p_{\mathrm{diff}}+p_{\mathrm{bind}}\}=0.96<1$, ensuring nonnegative transition probabilities and a column-stochastic transition matrix. The corresponding analytical equilibrium binding probability and expected equilibrium bound-molecule count are $h_{\mathrm{2D}}^{\mathrm{eq}}\approx0.95995$ and $\overline N_{\mathrm b}^{\mathrm{eq}}\approx959.95$, respectively.

To numerically verify the 2D construction, the analytical binding probabilities in~\cref{eq:supp_cir,eq:supp_local_response} are compared with a discrete-state particle-based Monte Carlo simulation (PBS) on the same state space. Rather than directly evaluating $\mathbf P_{\mathrm{2D}}^k$, each stochastic realization updates the molecular populations by multinomially sampling the destination counts of the molecules occupying each state. The reported PBS estimates are averaged over $N_{\mathrm{MC}}=100$ independent realizations and normalized by the release count $Q=1000$. Thus, the comparison verifies the consistency between the analytical matrix formulation and its stochastic implementation within the proposed Markov framework.

\Cref{fig:supp_results}(a) compares the total binding probability $h_{\mathrm{2D}}[k]$ predicted by the Markov formulation with the corresponding PBS estimate. The two results closely agree throughout the transient regime. The total binding probability initially increases as the released molecules diffuse from the central transmitter toward the reactive boundary and subsequently approaches the analytical equilibrium value $h_{\mathrm{2D}}^{\mathrm{eq}}$, indicated by the horizontal dashed line.

\Cref{fig:supp_results}(b) shows the local binding probabilities $h_{b_r}[k]$, $r\in\{1,2,3\}$, associated with the three bound states. Owing to the symmetry of the geometry about the vertical centerline, the responses of $b_1$ and $b_3$ overlap. The response of the central bound state $b_2$ is only slightly larger during the transient regime because it is directly aligned with the transmitter state $s_{3,3}$. However, the differences among the three responses remain small, indicating that lateral diffusion largely equalizes the arrival probabilities across the reactive sites. In the long-time limit, all three local responses converge to the same equilibrium value $h_{\mathrm{2D}}^{\mathrm{eq}}/3$, because the corresponding sites have identical binding and unbinding parameters.

\begin{figure}[t]
\centering
\subfigure[\label{fig:supp_total}Total binding probability]{%
\centering
\includegraphics[width=0.47\linewidth]{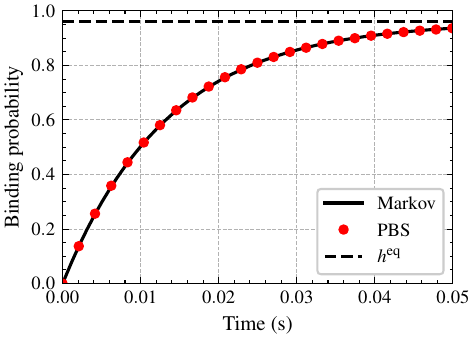}}
\hfill
\subfigure[\label{fig:supp_local}Local binding probabilities]{%
\centering
\includegraphics[width=0.47\linewidth]{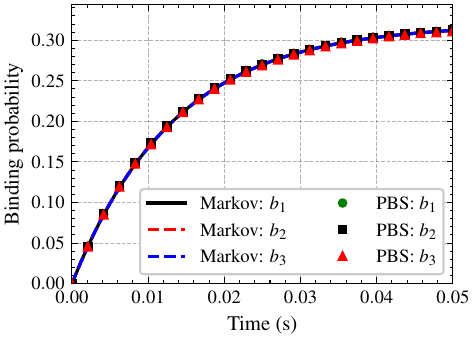}}
\caption{Numerical verification of the 2D Markov construction. Lines denote the analytical Markov predictions, whereas markers denote PBS estimates averaged over $N_{\mathrm{MC}}=100$ independent realizations. (a) Total binding probability $h_{\mathrm{2D}}[k]$. The horizontal dashed line indicates the analytical equilibrium probability $h_{\mathrm{2D}}^{\mathrm{eq}}$. (b) Local binding probabilities $h_{b_1}[k]$, $h_{b_2}[k]$, and $h_{b_3}[k]$. Owing to the symmetric geometry, the responses of $b_1$ and $b_3$ overlap, and all three responses converge to $h_{\mathrm{2D}}^{\mathrm{eq}}/3$.}
\label{fig:supp_results}
\end{figure}

\section{Numerical Validation of the Counting-Noise Covariance}

\label{sec:supp_cov_validation}

This section provides a direct numerical validation of the lag-dependent counting-noise covariance derived in \cref{prop:cov_microarray} of the main manuscript. While \cref{fig:rho_vs_ell} of the main manuscript characterizes the temporal correlation through the normalized correlation coefficient $\rho[\ell]$, here we directly compare the analytical covariance with Monte Carlo estimates without normalization.

For each random OOK sequence, the counting noise is obtained as $w[n]\triangleq z[n]-\bar z[n]$ (see \cref{eq:wk_def} in the main manuscript). The analytical covariance from \cref{prop:cov_microarray} and the corresponding Monte Carlo estimate are averaged over symbol indices after the initial transient and over independent random OOK sequences. The simulation parameters are consistent with those used for \cref{fig:rho_vs_ell} of the main manuscript.

\begin{figure}[htbp]
\centering
\subfigure[\label{fig:supp_cov_tb002}$T_b=0.02$ s]{%
\centering
\includegraphics[width=0.47\linewidth]{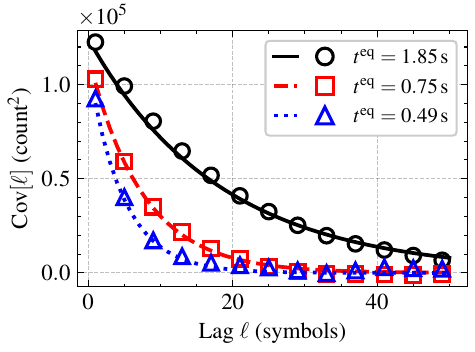}}
\hfill
\subfigure[\label{fig:supp_cov_tb005}$T_b=0.05$ s]{%
\centering
\includegraphics[width=0.47\linewidth]{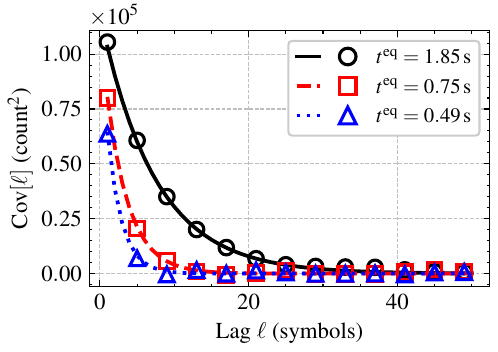}}

\subfigure[\label{fig:supp_cov_tb01}$T_b=0.1$ s]{%
\centering
\includegraphics[width=0.47\linewidth]{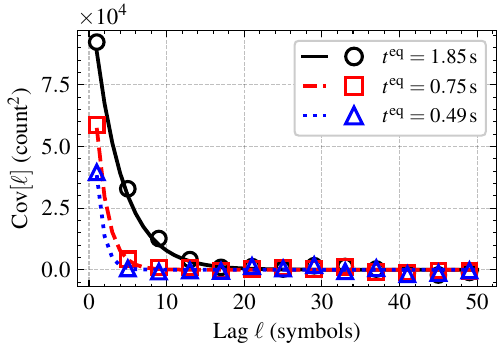}}
\hfill
\subfigure[\label{fig:supp_cov_tb03}$T_b=0.3$ s]{%
\centering
\includegraphics[width=0.47\linewidth]{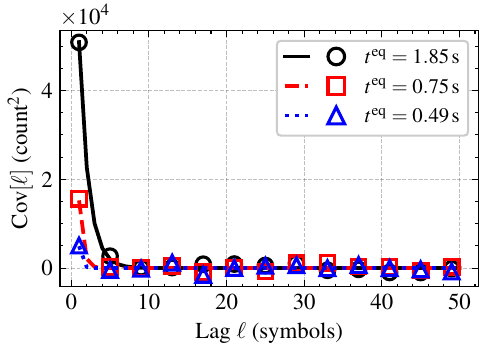}}

\caption{Analytical and Monte Carlo validation of the lag-dependent counting-noise covariance for different symbol intervals $T_b$. The curves show the covariance predicted by \cref{prop:cov_microarray}, while the markers denote Monte Carlo estimates. The three cases correspond to different settling times $t^{\mathrm{eq}}$, consistent with the settings used for the normalized correlation coefficient $\rho[\ell]$ in \cref{fig:rho_vs_ell} of the main manuscript.}
\label{fig:supp_cov_validation}
\end{figure}

As shown in~\cref{fig:supp_cov_validation}, the Monte Carlo estimates closely follow the analytical covariance over the considered lags and settling times. A longer settling time results in a stronger and more slowly decaying covariance across symbol intervals, while increasing $T_b$ substantially reduces the covariance at a given symbol lag. This trend is consistent with the normalized correlation results in \cref{fig:rho_vs_ell} of the main manuscript. These results directly validate both the magnitude and the lag-dependent behavior of the colored counting-noise covariance predicted by \cref{prop:cov_microarray}.

\section{Probe Capacity and Non-Saturation Assessment}
\label{sec:supp_probe_capacity}

This section quantifies the probe-availability margin for finite transmissions under the homogeneous, laterally uniform microarray model adopted in the main manuscript. The channel is initially empty, all functional probe sites are unoccupied, and each bound cDNA occupies one probe site.

\subsection{Surface Density and Total Probe Capacity}

In the proposed DNA microarray model, the effective volumetric probe concentration $c_{\mathrm p}$ is related to the probe surface density $\sigma_{\mathrm p}$ by $c_{\mathrm p}=\sigma_{\mathrm p}/\Delta x$, or equivalently, $\sigma_{\mathrm p}=c_{\mathrm p}\Delta x$. Using the parameters in \cref{tab:parameters} of the main manuscript, $c_{\mathrm p}=1\,\mu\mathrm M=10^{-3}\,\mathrm{mol\cdot m^{-3}}$ and $\Delta x=50\,\mathrm{nm}=5\times10^{-8}\,\mathrm m$, gives $\sigma_{\mathrm p}=5\times10^{-11}\,\mathrm{mol\,m^{-2}}$. This corresponds to approximately $30.1$ probe sites per $\mu\mathrm m^2$, or $3.01\times10^9$ sites per $\mathrm{cm}^2$.

For an effective reactive area $A_{\mathrm{RX}}$, the corresponding number of initially available probe sites is
\begin{equation}
N_{\mathrm p}=N_{\mathrm A}\sigma_{\mathrm p}A_{\mathrm{RX}},
\label{eq:supp_probe_capacity}
\end{equation}
where $N_{\mathrm A}=6.02\times10^{23}\,\mathrm{mol^{-1}}$ is Avogadro's constant. Hence, for a fixed probe surface density  $\sigma_{\mathrm p}$, the total probe capacity scales linearly with the effective reactive area $A_{\mathrm{RX}}$.

\subsection{Conservative Probe-Occupancy Assessment}

For a finite OOK sequence of $K$ symbols, with at most $Q$ molecules released per symbol, the cumulative number of released cDNAs is upper bounded by $QK$, corresponding to an all-one sequence. Since each bound cDNA occupies at most one probe site, the maximum fractional probe occupancy satisfies
\begin{equation}
\theta_{\max}\leq\frac{QK}{N_{\mathrm p}}
=\frac{QK}{N_{\mathrm A}\sigma_{\mathrm p}A_{\mathrm{RX}}}.
\label{eq:supp_probe_occupancy_bound}
\end{equation}

As an illustrative conservative example, consider $A_{\mathrm{RX}}=1\,\mathrm{cm}^2$, $Q=10^4$, and a transmission of $K=10^3$ symbols. The corresponding probe capacity is $N_{\mathrm p}\simeq3.01\times10^9$ sites. Even for an all-one sequence, which releases $QK=10^7$ cDNAs, \cref{eq:supp_probe_occupancy_bound} gives $\theta_{\max}\leq0.332\%$. Thus, more than $99.6\%$ of the probe sites remain available, illustrating that the considered operating conditions are well within the non-saturating regime.

The focus of this work is the communication behavior under the non-saturating probe approximation rather than the detailed geometrical design of a finite microarray. For a finite realization, a sufficiently large effective reactive area $A_{\mathrm{RX}}$ ensures that probe depletion remains negligible, as quantified by \cref{eq:supp_probe_occupancy_bound} and illustrated by the example above.

\section{Small-Scale Receiver-Level PBS Validation}
\label{sec:c4:validation}

To extend the PBS validation from the channel response to the receiver level, we conduct a \emph{small-scale} BER validation experiment comparing PBS with the analytical observation model derived from the proposed Markov framework. The purpose is to examine whether the receiver-performance predictions of the analytical model remain consistent with PBS under a representative operating condition.

\begin{table}[!b]
    \centering
    \caption{Small-scale receiver-level BER comparison between PBS and the analytical model.}
    \label{tab:c4:ber}
    \begin{tabular}{@{}lccccccc@{}}
        \toprule
        \multirow{2}{*}{\textbf{Receiver}} 
        & \multicolumn{2}{c}{\textbf{PBS}} 
        & \multicolumn{2}{c}{\textbf{Analytical}} 
        & \multirow{2}{*}{$\boldsymbol{\widehat{\delta}}$} 
        & \multirow{2}{*}{\textbf{95\% bootstrap CI}} \\
        \cmidrule(lr){2-3} \cmidrule(lr){4-5}
        & \textbf{BER} & \textbf{Errors} & \textbf{BER} & \textbf{Errors} &  &  \\
        \midrule
        Threshold   & 0.3262 & 1631 & 0.3274 & 1637 & $-0.0012$ & $[-0.0078,\,0.0054]$ \\
        DFE ($L=1$) & 0.1582 & 791  & 0.1544 & 772  & $0.0038$  & $[-0.0076,\,0.0156]$ \\
        DFE ($L=3$) & 0.0324 & 162  & 0.0320 & 160  & $0.0004$  & $[-0.0088,\,0.0098]$ \\
        DFE ($L=5$) & 0.0214 & 107  & 0.0218 & 109  & $-0.0004$ & $[-0.0078,\,0.0066]$ \\
        MAP-DF      & 0.0194 & 97   & 0.0202 & 101  & $-0.0008$ & $[-0.0078,\,0.0060]$ \\
        \bottomrule
    \end{tabular}
\end{table}

The physical parameters follow \cref{tab:parameters} of the main manuscript, with $Q=1000$ and $T_b=0.1\,\mathrm{s}$. For this validation, we perform $100$ independent Monte Carlo realizations, each consisting of $50$ independently generated equiprobable OOK bits, corresponding to $5000$ evaluated bits. Each realization starts from an empty channel, and all released molecules are propagated until the final observation. PBS and the analytical model use the same bit sequence and receiver parameters, with each feedback receiver using its own previous decisions. PBS explicitly propagates individual molecules according to the discrete channel dynamics, whereas the analytical model generates differential observations using the Markov-derived conditional mean and variance.

\Cref{tab:c4:ber} compares the resulting BERs. The BER difference is defined as $\delta=\mathrm{BER}_{\mathrm{PBS}}-\mathrm{BER}_{\mathrm{Analytical}}$, and its uncertainty is quantified by a $95\%$ paired-bootstrap confidence interval (CI) obtained by jointly resampling the $100$ matched Monte Carlo realizations. The largest absolute BER difference is $3.8\times10^{-3}$, and all paired $95\%$ confidence intervals include zero. Thus, this small-scale experiment provides a direct receiver-level PBS consistency check for the analytical model under the considered operating condition.

\end{document}